\documentclass{amsart}

\usepackage[babel,threshold=2]{csquotes}
\usepackage{amsfonts}
\usepackage{amssymb}
\usepackage{enumerate}
\usepackage{amsmath}
\usepackage{amsthm}
\usepackage{graphicx}
\usepackage[dvipsnames]{xcolor}
\usepackage[english]{babel}
\usepackage{comment}
\usepackage{tensor}
\usepackage{hyperref}
\usepackage{multirow}
\usepackage{float}
\usepackage[normalem]{ulem}

\numberwithin{equation}{section}

\newtheorem{theorem}{Theorem}[section]

\newtheorem{claim}[theorem]{Claim}

\newtheorem{corollary}[theorem]{Corollary}

\newtheorem{fact}[theorem]{Fact}

\newtheorem*{T:equiv}{Theorem~\ref{t:equiv}}
\newtheorem*{C:prec}{Corollary~\ref{c:prec}}

\theoremstyle{definition}

\newtheorem{remark}[theorem]{Remark}
\newtheorem{definition}[theorem]{Definition}

\DeclareMathOperator{\dH}{d_H}

\DeclareMathOperator{\Tr}{Tr}

\newcommand{\NN}{\mathbb{N}}

\newcommand{\RR}{\mathbb{R}}

\newcommand{\CC}{\mathbb{C}}

\newcommand{\iH}{\mathcal{H}}

\newcommand{\iS}{\mathcal{S}}
\newcommand{\iN}{\mathcal{N}}

\newcommand*{\defeq}{\stackrel{\text{def}}{=}}

\renewcommand{\phi}{\varphi}
\renewcommand{\epsilon}{\varepsilon}

\begin{document}

\title[Effective eigenvalue approximation from moments]{Effective eigenvalue approximation from moments for self-adjoint trace-class operators}

\author{Rich\'ard Balka} 
\address{HUN-REN Alfr\'ed R\'enyi Institute of Mathematics, Re\'altanoda u.~13--15, H-1053 Budapest, Hungary, AND Institute of Mathematics and Informatics, Eszterh\'azy K\'aroly Catholic University, Le\'anyka u.~4, H-3300 Eger, Hungary}
\email{balkaricsi@gmail.com}

\author{G\'abor Homa}
\address{HUN-REN Wigner Research Centre for Physics, P.~O.~Box 49, H-1525 Budapest, Hungary}
\email{homa.gabor@wigner.hun-ren.hu}

\author{Andr\'as Csord\'as}
\address{Department of Physics of Complex Systems, E\"{o}tv\"{o}s Lor\'and University, P\'azm\'any P\'eter s\'et\'any 1/A, H-1117 Budapest, Hungary}
\email{csordas@tristan.elte.hu}



\begin{abstract} Spectral properties of bounded linear operators play a crucial role in several areas of mathematics and physics. For each self-adjoint, trace-class operator $O$ we define a set $\Lambda_n\subset \RR$, and we show that it converges to the spectrum of $O$ in the Hausdorff metric under mild conditions. Our set $\Lambda_n$ only depends on the first $n$ moments of $O$. We show that it can be effectively calculated for physically relevant operators, and it approximates the spectrum well without diagonalization. 

We prove that using the above method we can converge to the minimal and maximal eigenvalues with super-exponential speed. 

We also construct monotone increasing lower bounds $q_n$ for the minimal eigenvalue (or decreasing upper bounds for the maximal eigenvalue). This sequence only depends on the moments of $O$ and a concrete upper estimate of its $1$-norm; we also demonstrate that $q_n$ can be effectively calculated for a large class of physically relevant operators. This rigorous lower bound $q_n$ tends to the minimal eigenvalue with super-exponential speed provided that $O$ is not positive semidefinite. As a by-product, we obtain computable upper bounds for the $1$-norm of $O$, too. 

Numerical examples demonstrate the relevance of our approximation in estimating entropy and negativity, which is useful, among others, in quantum optical and in open quantum system models. The results can be directly applicable to problems in quantum information, statistical mechanics, and quantum thermodynamics, where using traditional techniques based on diagonalization is impractical.

\end{abstract}

\maketitle

\section{Introduction}
\subsection{History and motivation}
Self-adjoint trace-class linear operators on infinite 
dimensional separable Hilbert spaces played an inevitable role in the foundation of quantum mechanics \cite{Neumann,Weyl,Wigner,Husimi,Moyal}. In particular, density operators \cite{Neumann} have found many applications in quantum chemistry, statistical mechanics, quantum optics, quantum information theory, and the theory of open quantum systems \cite{Hillery, Lee, Schleich, Weinbub,Salazar,book1,HPZ,H-Y}.
These operators occur in mathematical physics \cite{grothendieck1966produits,dubinsky1979structure,wong1979schwartz}, and they recently appeared in the modern physical description of our universe \cite{Fahn_2023,hsiang2024graviton}.  

Finding the spectrum of linear operators, or deciding whether they are at least positive semidefinite is an important and delicate problem: positivity preservation during the time evolution is a crucial question in the dynamics of open quantum systems \cite{Gnutzmann,Gosson,Newton,BLH,HBCSCS,HCSB}.
Quantum entanglement problems still pose hard challenges to mathematicians and theoretical physicists \cite{EPR,Bell,HORODECKI19961,Open}; they are also intimately connected to positivity via the Peres--Horodecki criterion \cite{HORODECKI19961,Perescrit}. Describing entanglement in a quantitative way is a very important task in the theory of quantum computing, and quantum communication. Therefore, it is useful to know the magnitude of the \emph{negativity}\footnote{The sum of the absolute values of the negative eigenvalues.} of a density operator after the partial transpose \cite{Duan,Ruggiero,Kumari,Kao_2016}. While the spectrum can be calculated relatively easily in finite dimensional Hilbert spaces \cite{separable,Eisert,Vidal,Plenio}, for infinite dimensional non-Gaussian operators it is still a hard problem \cite{Lami_2023,Pirandola,Vogel,Miki}. Infinite dimensional trace-class integral operators arise in the modeling of quantum optical channels, the theory of distributed quantum communication, and the quantum internet \cite{Sharma2023,pathak2013elements,Kao_2016,Rohde_2021}.

The spectrum of a trace-class operator determines a lot of important quantities, arguably the most important and difficult ones are the entropies, for example the R\'enyi or von Neumann entropy, see \cite{Shirokov2010}. Entropy is a crucial notion in the areas of quantum information theory \cite{nielsen_chuang_2010}, quantum communication \cite{Sharma2023}, quantum computation \cite{pathak2013elements}, quantum cryptography \cite{Yao_22}, quantum technology of entanglement distillation and noisy quantum channels \cite{Lami2024} just to name a few. The eigenvalues and kernels of self-adjoint trace-class integral operators has also many applications in the area of machine learning, see e.\,g.~\cite{DBL} and references therein. As the eigenvalues and eigenvectors of a density operator encode all information about the physical subsystem it describes \cite{Neumann,Neumann1927,Lloyd2014}, approximating the spectrum has far-reaching physical applications and mathematical relevance. 

One of the main challenges in quantum technology is to filter out quantum noise stemming from the imprecise quantum physical measurements and other external circumstances \cite{clerk2010introduction}. Mathematical and physical properties of quantum noise and quantum decoherence are only understood in simple models so far \cite{chang2015quantum,beny2009quantum}. Understanding the dynamics of entanglement between two coupled Brownian oscillators (considered as an open quantum system) is a notoriously hard problem \cite{entang_Hu}. Here time evolution is described by master equations which are very hard to solve analytically, and the problem is usually formalized in infinite dimensional Hilbert spaces, which also adds to the challenge. The Gaussian case can be solved exactly under certain conditions \cite{entang_Hu}, but the non-Gaussian case is still mostly unexplored. Here approximating the eigenvalues yields estimate for the fidelity \cite{nielsen_chuang_2010} and for various quantum entropies \cite{ohya2004quantum}.

The \emph{maximal eigenvalue} also has a meaning in statistical physics as follows. We consider a Hamiltonian operator $\hat{H}$ of a physical system, whose energy spectrum is supposed to have a lower bound. The eigenstate with the lowest energy is called the \emph{ground state}. In many cases, after a physical measurement, the examined physical system can be found in the ground state with larger probability than in any other state. Also, the largest probability equals the maximal eigenvalue of the statistical physical density operator. For example, in case of canonical ensembles in statistical quantum physics, the equilibrium state is characterized by a density operator $\hat{\rho}=e^{-\beta \hat{H}}/\Tr \big\{e^{-\beta \hat{H}}\big\}$, see \cite{reif2009fundamentals}. Here the maximal eigenvalue equals the probability of being in the ground state; this probability is an important quantity e.\,g.~in the quantum thermodynamic description of physical systems \cite{diosi2011short}. 

The \emph{minimal eigenvalue} of a self-adjoint operator, if negative, is a quantitative measure how far the operator is from being positive semidefinite. 

Estimating the Schatten--von Neumann $1$-norm of trace-class integral operators plays an important role in the mathematics of anharmonic oscillators \cite{DELGADO20211} and areas of harmonic analysis related to physics \cite{Catana}.

One of our main examples will be \emph{polynomial Gaussian} operators, see Definition~\ref{d:poly}. For the physical relevance of polynomial Gaussian operators and for qualitative results about their positivity the readers might consult~\cite{balka2024positivity}. 

 Self-adjoint trace-class operators  are not only used in physics, but in related fields as well \cite{Haferkamp2024}. The kernels of these integral operators are given by polynomial times Gaussian kernels, which occur, for example, in general quantum physics (see, for example, the article \cite{poly_Gauss_qphys1}  and the references therein), quantum optics \cite{Szabo}, physics of open quantum systems \cite{Fleming,HHBACS}, in nonequilibrium statistical physics \cite{Dolgirev}, high energy nuclear physics \cite{Rais2025}, and at last, but not least machine learning and related applications  \cite{Polygauss_estimation}.

Finally, we believe that our results can be extended to more general cases, and thus we can get closer to answering many theoretical and experimental physics questions. For example, an interesting problem in  case of infinite dimensional Hilbert spaces is to calculate the entanglement related quantities for real physical systems \cite{HHBACS}. The method of diagonalization is often too involved to extract spectral quantities or it might not work at all. Instead, we provide a more general method for which it is enough to calculate the different moments of our operator. 

%When determining in this case these quantities, it is necessary to be able to calculate the spectral characteristics, which is not an easy task. Our work has significant benefits in this research area, and here we  demonstrate this with examples.

Our results offer a powerful and efficient alternative to traditional spectral methods in various areas of physics, most importantly quantum physics. In particular, they open the way for computing quantum entropies and entanglement measures \cite{Szalay,szalay-toth} in systems where diagonalization does not work for theoretical or practical reasons; this includes non-Gaussian infinite dimensional models that arise naturally in quantum optics and open quantum systems.

\subsection{Physical example: Entanglement dynamics of two coupled quantum harmonic oscillators}\label{ss:physics_example} 

 As a motivation, in this subsection we consider the time evolution of a relatively simple physical system, for which 
%it is important to compute 
entanglement measures such as logarithmic negativity \cite{Vidal,Plenio} and mutual information   \cite{DeTomasi2017,HHBACS} show nontrivial behavior from the point of view of quantum-information. Mutual information is derived from the von~Neumann entropy, while logarithmic negativity is obtained from the Schatten--von~Neumann norm of the partially transposed density operator \cite{book1,simon,SimonR}. The computation of both of these quantities requires the knowledge of the whole spectrum of the corresponding self-adjoint trace-class integral operator. However, most of the time it is almost impossible to determine the exact eigenvalues. Hence this serves as a strong physical motivation to approximate the spectrum of operators.

We examine a system comprising two coupled harmonic oscillators with different masses. The Hamiltonian operator for this system is expressed as
\begin{equation}
\hat{H}_\text{sys}=\frac{\hat{P}_1^2}{2m_1}+\frac{1}{2}m_1\Omega^2\hat{x}_1^2+\frac{\hat{P}_2^2}{2m_2}+\frac{1}{2}m_2\Omega^2\hat{x}_2^2+\kappa(\hat{x}_1-\hat{x}_2)^2, 
\end{equation}
where $\hat{x}_i$ and $\hat{P}_i$ are position and momentum operators for the mechanical oscillators with masses $m_1$ and $m_2$ which oscillate at the same frequency $\Omega$, respectively. The coupling constant between the two oscillators is denoted by $\kappa$. We can introduce the following notations for quantities interpreted in  center of mass and relative coordinate systems
\begin{align*}
\hat{X}:=&\frac{m_1 \hat{x}_1 + m_2 \hat{x}_2}{m_1+m_2},~ \hat{P}:=\hat{P}_1+\hat{P}_2,~M:=m_1+m_2 \\
\hat{x}:=&\hat{x}_1-\hat{x}_2,~\hat{p}:=\frac{m_2}{m_1+m_2}\hat{P}_1-\frac{m_1}{m_1+m_2}\hat{P}_2,~\mu:=\frac{m_1 m_2}{m_1+m_2},
\end{align*} 
thus, the Hamiltonian with the new operators has the form
\begin{equation}
\hat{H}_{\text{sys}}=\frac{\hat{P}^2}{2M}+\frac{1}{2}M\Omega^2\hat{X}^2+\frac{\hat{p}^2}{2\mu}+\frac{1}{2}\mu\Omega^2\hat{x}^2+\kappa\hat{x}^2. 
\end{equation}
The state of this system evolves according to the Neumann equation
\begin{equation*}
    \frac{d\hat{\rho}_{\text{sys}}(t)}{dt}=-\frac{i}{\hbar}\left [\hat{H}_\text{sys},\hat{\rho}_{\text{sys}}(t)\right].
\end{equation*}
In the next step we consider two representations
\begin{align*}
\rho_o(\Delta_1,k_1,\Delta_2,k_2, t)&=\mathrm{Tr} \big\{\hat{\rho}(t) \exp\{ik_1 \hat{x}_1- i \Delta_1 \hat{P}_1+ik_2 \hat{x}_2- i \Delta_2 \hat{P}_2\} \big\} \text{ and } \\
\rho_c(\Delta,K,\delta,k, t)&= \mathrm{Tr} \big\{\hat{\rho}(t) \exp\{ik \hat{x}- i \delta \hat{p}+iK \hat{X}- i \Delta \hat{P}\} \big\},
\end{align*}
where $\rho_o$ and $\rho_c$ are the characteristic functions in the original variables and in the center of mass and relative variables, respectively. 
Here we have
\begin{eqnarray}
		\label{eq:transformation_1}
		K&:=&k_1+k_2, \quad k:=\frac{m_2}{m_1+m_2}k_1-\frac{m_1}{m_1+m_2}k_2, \nonumber\\ 
        \Delta &:=&\frac{m_1 \Delta_1+m_2\Delta_2}{m_1+m_2}, \quad \delta:=\Delta_1-\Delta_2, \nonumber
\end{eqnarray}
and finally,  after the necessary transformations and calculations we get the equation for the time evolution of the characteristic function
\begin{equation}
  \label{eq:master_equation}
			\frac{\partial}{\partial t} \rho_c(\mathbf{w},t) =  \biggl[-\frac{\hbar K}{M} \frac{\partial}{\partial \Delta}-\frac{\hbar k}{\mu} \frac{\partial}{\partial \delta}+ 
			\frac{M \Omega^2}{\hbar}  \Delta\frac{\partial}{\partial K}+ \frac{\mu\Omega^2+2\kappa}{\hbar} \delta \frac{\partial}{\partial k}    			\biggr]\rho_c(\mathbf{w},t),  
\end{equation}
where 
\begin{equation*}
    \mathbf{w}^\top=(\Delta, K, \delta, k).
\end{equation*}
Eq.~\eqref{eq:master_equation} must be solved with initial condition that at $t=0$ the $\rho$ phase space function is given by a prescribed $\rho_{c_0}$ function
\begin{equation*} \label{eq:initial_value}
\rho_c(\mathbf{w},0)= \rho_{c_0} (\mathbf{w})  .
 \end{equation*}
We note that $\mathrm{Tr}\,\hat{\rho}=1$ requires
\begin{equation*}
 \rho_{c_0} (\mathbf{w}=\mathbf{0})=1.
\end{equation*}
Let us introduce the following quantity:
 \begin{equation*}
 \label{eq: new_frequencies}
\quad \omega^2_d:=\Omega^2+\frac{2 \kappa}{\mu},
 \end{equation*}
then \eqref{eq:master_equation} can be written as
\begin{equation*}
 \frac{\partial \rho_c}{\partial t}  +\frac{\hbar K}{M}\frac{\partial \rho_c}{\partial \Delta}-\frac{M \Omega^2 \Delta}{\hbar} \frac{\partial \rho_c}{\partial K}+\frac{\hbar k}{\mu}\frac{\partial \rho_c}{\partial \delta}-\frac{\mu\omega_d^2 \delta}{\hbar} \frac{\partial \rho_c}{\partial k}=0.
 \label{eq:master_eq_rewritten}
\end{equation*}
The left-hand side is a total derivative if one identifies the following equations for the characteristics:
\begin{equation}
\frac{ d\Delta}{dt}=\frac{\hbar K}{M}, \quad \frac{ d K}{dt}=-\frac{M\Omega^2 \Delta}{\hbar}, \quad \frac{ d\delta}{dt}= \frac{\hbar k}{\mu}, \quad \frac{ d k}{dt}=-\frac{\mu\omega_d^2 \delta}{\hbar}.
\label{eq:ODE_for_characterteristics}
\end{equation}
The equations (\ref{eq:ODE_for_characterteristics}) are linear in the phase-space variables with a constant coefficient matrix 
 \begin{equation}
 \label{eq:characteristic_equation}
\frac{d}{dt}  \begin {pmatrix}\Delta\\ \noalign{\medskip}K
\\ \noalign{\medskip}\delta\\ \noalign{\medskip}k\end {pmatrix}=     \begin {pmatrix} 0 &\frac{\hbar}{M}&0&0\\ \noalign{\medskip}
-\frac{M{\Omega}^{2}}{\hbar}&0&0&0\\ \noalign{\medskip}0&0&0&\frac{ \hbar}{\mu}
\\ \noalign{\medskip}0&0&-\frac{\mu{\omega_d}^{2}}{\hbar}&0
\end {pmatrix} 
          \begin {pmatrix}\Delta\\ \noalign{\medskip}K
\\ \noalign{\medskip}\delta\\ \noalign{\medskip}k\end {pmatrix}\equiv \mathbf{M} \cdot\begin{pmatrix} \Delta \\K \\ \delta \\k \end{pmatrix}.
 \end{equation}
The solution of \eqref{eq:characteristic_equation} is
\begin{equation*}
\mathbf{w}(t)=\exp{(\mathbf{M} t)}\cdot \mathbf{w}(0).
\label{eq:evolution_of characteristics}
\end{equation*}
%The initial condition is fulfilled if
%\begin{equation*}
%\rho_c(t=0) \equiv \rho_c(\mathbf{w}(0))=\rho_{c_0}(\mathbf{w}(0)),
%\end{equation*}
It can be checked that the solution of \eqref{eq:master_equation} is
\begin{equation}  \label{eq:solution}
\rho_c(\mathbf{w},t)=\rho_{c0}\Bigl(\exp(-\mathbf{M} t) \cdot \mathbf{w}\Bigr).\end{equation}
The initial condition is trivially fulfilled.

We note that the characteristic function can be obtained as the Fourier transform of the Wigner function in all of its  phase-space variables. Consequently, the solution also preserves the mathematical structure of the initial density operator in the position representation as well. So, if we start the time evolution of a density operator from a polynomial Gaussian  form,  it will remains in this form at all times, but with time-dependent parameters. After partial transposition, our approximation method described above can provide a good estimate of the logarithmic negativity even for non-trivial initial conditions. We can only mention as a non-trivial example that the mutual information between individual physical subsystems can be calculated for more general operator families with external noises, as in the article \cite{HHBACS}.

\subsection{Summary of our results and the structure of the paper} Let $\iH$ be a separable, complex Hilbert space and let $O\colon \iH \to \iH$ be a self-adjoint, trace-class linear operator. Then $O$ has only countably many real eigenvalues, let us enumerate them with multiplicities as $\{\lambda_i\}_{i\geq 0}$. As $O$ is trace-class, we have $\sum_{i\geq 0} |\lambda_i|<\infty$. If there are only finitely many eigenvalues $\lambda_0,\dots,\lambda_m$ then define $\lambda_i=0$ for each integer $i>m$.

Following \cite{Newton} we define symmetric functions of the eigenvalues $e_k$ such that $e_0=1$ and for $k\geq 1$ we have
\begin{equation*}
e_k\defeq \sum_{0\leq i_1<\dots<i_k} \lambda_{i_1}\cdots \lambda_{i_k}; 
\end{equation*}
we can calculate $e_k$ from the moments of $O$, see Subsection~\ref{ss:ek}. In Section~\ref{s:spectrum} we introduce the finite set $\Lambda_n\subset \RR$ as 
\begin{equation*}
\Lambda_n\defeq \left\{-\frac{1}{x}: x\in \RR \text{ and } \sum_{k=0}^n e_kx^k=0\right\}.
\end{equation*}
In order to find $\Lambda_n$, we only need to calculate the above mentioned quantities $e_k$ for $k=0,\dots,n$.
We prove in Theorem~\ref{t:conv} that $\Lambda_n$ converges to the closure of the set of non-zero eigenvalues in the Hausdorff metric, see Subsection~\ref{ss:theor}. In Subsection~\ref{ss:EGauss} we demonstrate by a Gaussian example that $\Lambda_n$ is very close to the set of eigenvalues even for $n=10$; in this case the exact spectrum is known and can be compared to $\Lambda_n$. 
 In Subsection~\ref{ss:NGauss} we calculate $\Lambda_{10}$ for a polynomial Gaussian operator (see Definition~\ref{d:poly}), where the exact spectrum is unknown; we also approximate the spectrum, the Schatten--von Neumann norm, and the von Neumann entropy in the positive semidefinite case.

In the latter sections we approximate the extremal eigenvalues and verify faster convergence. We only consider the case of the minimal eigenvalue, since replacing the operator $O$ with $-O$ handles the maximal eigenvalue as well. Let us define  
\begin{equation*} 
\lambda_{\min}\defeq \min\{0,\lambda_i: i\geq 0\},
\end{equation*} which is zero if $O$ is positive semidefinite, and equals to the \emph{minimal eigenvalue}\footnote{$\lambda_{\min}$ is the infimum of the eigenvalues except for the case when $O$ has only finitely many strictly positive eigenvalues; we are clearly more interested in the infinite dimensional case though.} smaller than zero otherwise. 

In Section~\ref{s:qn0} we approximate $\lambda_{\min}$ by the sequence $\{q_{n,0}\}_{n\geq 1}$, where 
\begin{equation*} 
q_{n,0}\defeq \frac{-1}{\min \left\{x>0: \sum_{k=0}^n e_kx^k=0\right\}}.
\end{equation*} 
%Therefore, we give another estimate $q_{n,0}$ of $\lambda_{\min}$ using also $e_k$ in Section~\ref{s:qn0}. On the bad side, $q_{n,0}$ is not necessarily monotone anymore. On the good side, it seems to approximate $\lambda_{\min}$ much better than $q_{n,c}$, and the main reason of this is that $q_{n,0}$ does not depend on $c$ through the later sub-optimal estimate \eqref{eq:ek}. 
In Theorem~\ref{t:c=0} we prove that $q_{n,0}\to \lambda_{\min}$ with super-exponential speed:
\begin{equation*} |q_{n,0}-\lambda_{\min}|\leq \exp\left[-n \log(n)+\mathcal{O}(n)\right]
\end{equation*} 
under the slight technical condition that $\lambda_{\min}$ is negative and has multiplicity one.

In Section~\ref{s:lmin} we assume that $c>0$ is a number satisfying
\begin{equation} \label{eq:cc}
\sum_{i\geq 0} |\lambda_i|\leq c.
\end{equation}  
Then we can construct and calculate a sequence $\{q_{n,c}\}_{n\geq 0}$ such that
\begin{enumerate}[(i)]
\item \label{i:q1} $q_{n,c}$ is monotone increasing and $q_{n,c}\to \lambda_{\min}$ as $n\to \infty$;
\item \label{i:q2} $q_{n,c}\approx -\frac cn$ if $O$ is positive semidefinite;
\item \label{i:q3} $q_{n,c}\to \lambda_{\min}$ with super-exponential speed if $O$ is not positive semidefinite. 
\end{enumerate}
In Section~\ref{s:Sch} we demonstrate that we can effectively find values $c$ according to \eqref{eq:cc} for a large class of integral operators having \emph{polynomial Gaussian} kernels, see Definition~\ref{d:poly} for the precise notion. Finding an upper bound for the \emph{Schatten--von Neumann $1$-norm} $\sum_{i\geq 0} |\lambda_i|$ is a hard problem in itself. We do this by writing $O$ as the product of two, typically not self-adjoint Hilbert--Schmidt integral operators in different ways, and we apply H\"older's inequality \eqref{eq:OHolder}; finally, we are even able to optimize the result. This upper bound gives an automatic lower bound for the minimal eigenvalue and an upper bound for the negativity as well, see inequality \eqref{eq:Trmin} and its practical calculation in \eqref{eq:lminbound}. In Subsections~\ref{ss:Gauss} and \ref{ss:quad} we show this method in case of a Gaussian kernel, and a Gaussian kernel multiplied by a quadratic polynomial, respectively. 

In Subsection~\ref{ss:qnc} we define $q_{n,c}$ using the above defined quantities $e_k$ and $c$. We prove \eqref{i:q1} in Theorem~\ref{t:qmon}. This means that the sequence $q_{n,c}$ is monotone increasing and converges to $\lambda_{\min}$ no matter what the value of $c$ is, getting better and better, rigorous lower bounds during the process. We show \eqref{i:q2} with precise lower and upper bounds in Theorem~\ref{t:qc}, implying basically that $q_{n,c}$ depends on $c$ in a linear way if $O$ is positive semidefinite.  
For \eqref{i:q3} see Theorem~\ref{t:qexp}, where we prove a super-exponential speed of convergence 
\begin{equation*}
|q_{n,c}-\lambda_{\min}|\leq 
\exp\left[-\frac{|\lambda_{\min}|}{c} n \log(n)+\mathcal{O}(n)\right]
\end{equation*}
if $O$ is not positive semidefinite. We also calculate and compare our estimates $q_{n,0}$ and $q_{n,c}$ for a concrete operator in Figure~\ref{fig:nongaussian_qns}, obtaining that the sequence $q_{n,0}$ converges extremely fast in practice.

\section{Preliminaries}
Let $\iH$ be a separable, complex Hilbert space\footnote{We use the mathematical convention $\langle \lambda f, g\rangle=\lambda \langle f, g\rangle $ and $\langle f, \lambda g\rangle=\lambda^{*} \langle f, g\rangle$, where $\lambda^{*}$ denotes the conjugate of $\lambda$.} and let $O\colon \iH\to \iH$ be a bounded linear operator. We say that $\lambda\in \CC$ is an \emph{eigenvalue} of $O$ if there exists a nonzero $f\in \iH$ such that $Of=\lambda f$; then $f$ is an \emph{eigenvector} corresponding to $\lambda$. The set of eigenvectors 
\begin{equation*} 
E_{\lambda}=\{f\in \iH: Of=\lambda f\} 
\end{equation*}
is called the \emph{eigenspace} corresponding to $\lambda$. 
We call $O$ \emph{self-adjoint} if $O^\dagger=O$, which implies that all eigenvalues are real. We say that $O$ is \emph{positive semidefinite} if $\langle Of, f\rangle \geq 0$ for all $f\in \iH$. A positive semidefinite operator is always self-adjoint. We call $O$ \emph{Hilbert--Schmidt} if there exists an orthonormal basis $\{f_i\}_{i\geq 0}$\footnote{The notation $\{f_i\}_{i\geq 0}$ and $\{\lambda_i\}_{i\geq 0}$ wants to take into account that there might be only finitely many basis vectors $f_i$ and eigenvalues $\lambda_i$, respectively.} of $\iH$ such that $\sum_{i\geq 0} \langle Of_i, Of_i \rangle<\infty$. We say that $O$ is \emph{trace-class}\footnote{Linear operators satisfy the inclusion: $\text{trace-class} \subset \text{Hilbert--Schmidt}\subset \text{compact} \subset \text{bounded}$.} if there is an orthonormal basis $\{f_i\}_{i\geq 0}$ of $\iH$ such that $\sum_{i\geq 0} |\langle |O|f_i, f_i \rangle |<\infty$, where $|O|\defeq \sqrt{O^{\dagger}O}$. Then the \emph{trace} of $O$ is defined as 
\begin{equation*}
\Tr\{O\}\defeq \sum_{i\geq 0} \langle Of_i, f_i \rangle, 
\end{equation*} 
where the sum is absolutely convergent and independent of the orthonormal basis $f_i$, see \cite[Definition~9.3]{deGosson_book}.
Now assume that $O$ is self-adjoint and trace-class.
Every trace-class operator is compact by \cite[Proposition~9.7]{deGosson_book}. As $O$ is compact and self-adjoint, the spectral theorem \cite[Theorem~3 in Chapter~28]{lax2002functional} implies that it has countably many eigenvalues\footnote{Note that in the enumeration $\{\lambda_i\}_{i\geq 0}$ we list every eigenvalue $\lambda$ with multiplicity $\dim E_{\lambda}$.} $\{\lambda_i\}_{i\geq 0}$, and there exists an orthonormal basis consisting of eigenvectors. Choosing such a basis yields  
\begin{equation*} 
 \Tr \{O\}=\sum_{i\geq 0} \lambda_i. 
\end{equation*}
The main example in this paper is $\iH=L^2(\mathbb{R}^{d})$, that is, the Hilbert space of the complex-valued square integrable functions defined on $\RR^d$ endowed with the scalar product $\langle f,g \rangle=\int_{\RR^d} f(\mathbf{x})g^{*}(\mathbf{x}) \, \mathrm{d} \mathbf{x}$, where $z^{*}$ denotes the complex conjugate of $z$. A kernel $K \in L^2(\mathbb{R}^{2d})$ defines an integral operator $\widehat{K}\colon L^2(\mathbb{R}^{d})\to L^2(\mathbb{R}^{d})$ by the formula
\begin{equation*} 
 \big( \widehat{K} f \big) (\mathbf{x}) = \int_{\mathbb{R}^{d}}\, K(\mathbf{x},\mathbf{y}) f(\mathbf{y})\, \mathrm{d}\mathbf{y}. 
\end{equation*}
The operator $\widehat{K}$ is always Hilbert--Schmidt, so compact, see e.\,g.~\cite{deGosson_book}. Thus $\widehat{K}$ has countably many eigenvalues $\{\lambda_i\}_{i\geq 0}$, and each eigenvalue $\lambda$ satisfies a Fredholm integral equation
\begin{equation*} 
\int_{\RR^d} \, K(\mathbf{x},\mathbf{y}) f(\mathbf{y})\,\mathrm{d}\mathbf{y}=\lambda f(\mathbf{x})
\end{equation*} 
with some nonzero functions $f\in L^2(\mathbb{R}^d)$. 
If the kernel $K$ is continuous, then $\widehat{K}$ is self-adjoint if and only if $K(\mathbf{y},\mathbf{x})=K^*(\mathbf{x},\mathbf{y})$ for all $\mathbf{x},\mathbf{y}\in \RR^d$. It is a trickier question when $\widehat{K}$ becomes trace-class; somewhat surprisingly adding the natural-looking condition $K\in L^1(\RR^{2d})$ is not sufficient, see \cite[Subsection~9.2.2]{deGosson_book}. The \emph{Schwartz space} $\iS(\RR^d)$ is the set of smooth functions $f\colon \RR^d\to \CC$ with rapidly decreasing mixed partial derivatives, see e.\,g.~\cite[Section~V.3]{SR}. If $K\in \iS(\RR^{2d})$, then $\widehat{K}$ is trace-class, see \cite[Proposition~287]{Gosson_Harmonic} or \cite[Proposition~1.1]{Brislawn} with the remark afterwards. If $K$ is continuous and $\widehat{K}$ is trace-class, then \cite{Duflo,Brislawn} yield the useful formula
\begin{equation*} 
\Tr \big\{\widehat{K}\big\}=\int_{\RR^d} K(\mathbf{x},\mathbf{x}) \, \mathrm{d}\mathbf{x}.
\end{equation*} 
Note that $\widehat{K}$ is positive semidefinite if and only if all of its eigenvalues are greater than or equal to zero, which is equivalent to
\begin{equation*} 
\int_{\mathbb{R}^{d}} \int_{\mathbb{R}^{d}} K(\mathbf{x},\mathbf{y}) f(\mathbf{x})f^{*}(\mathbf{y})  \, \mathrm{d} \mathbf{x} \, \mathrm{d} \mathbf{y} \geq 0 \quad \text{for all } f \in  L^2(\mathbb{R}^d).
\end{equation*}
For more on Hilbert--Schmidt and trace-class operators see e.\,g.~\cite{Gosson_Harmonic,deGosson_book,Stein}. 

We say that $\widehat{\rho}$ is a \emph{density operator} if it is positive semidefinite with $\Tr \big\{\widehat{\rho}\big\}=1$, that is, we have $\lambda_i \geq 0$ and $\sum_{i\geq 0} \lambda_i=1$.

\begin{definition} \label{d:poly}
A \emph{Gaussian} kernel $K_G\in L^2(\RR^{2d})$ is of the form 
\begin{equation*} K_G(\mathbf{x},\mathbf{y})=\exp\left\{-\mathbf{r}^{\top}\mathbf{M}\mathbf{r}- \mathbf{V}^{\top}\mathbf{r} + F\right\},
\end{equation*}
where $\mathbf{r}=(\mathbf{x},\mathbf{y})^{\top}$, and $\mathbf{M}=\mathbf{M}_1+i\mathbf{M}_2\in \CC^{2d\times 2d}$ is a symmetric complex matrix such that $\mathbf{M}_1,\mathbf{M}_2\in \RR^{2d\times 2d}$ and $\mathbf{M}_1$ is positive definite, $\mathbf{V}\in \CC^{2d}$ is a complex vector, and $F\in \CC$ is a complex number. In particular, a \emph{self-adjoint Gaussian} kernel $K_G\in L^2(\RR^{2d})$ is of the form 
\begin{align*}
K_G(\mathbf{x},\mathbf{y})=&\exp \left\{-(\mathbf{x}-\mathbf{y})^{\top} \mathbf{A} (\mathbf{x}-\mathbf{y}) -i(\mathbf{x}-\mathbf{y})^{\top} \mathbf{B} (\mathbf{x}+\mathbf{y})\right.
\\
&\quad \quad \left. -(\mathbf{x}+\mathbf{y})^{\top} \mathbf{C}(\mathbf{x}+\mathbf{y})-i\mathbf{D}^{\top}(\mathbf{x}-\mathbf{y})-\mathbf{E}^{\top}(\mathbf{x}+\mathbf{y})+F\right\},
\end{align*} 
where $\mathbf{A},\mathbf{B},\mathbf{C}\in \RR^{d\times d}$ are real  matrices such that $\mathbf{A},\mathbf{C}$ are symmetric and positive definite, $\mathbf{E},\mathbf{D}\in \RR^d$, and $F\in \RR$. 
We call $K\in L^2(\RR^{2d})$ a \emph{polynomial Gaussian} kernel if
\begin{equation} \label{eq:poly}
K(\mathbf{x},\mathbf{y})=P(\mathbf{x},\mathbf{y})K_G(\mathbf{x},\mathbf{y}), 
\end{equation} 
where $P$ is a self-adjoint polynomial in $2d$ variables, and $K_G\in L^2(\RR^{2d})$ is a self-adjoint Gaussian kernel.
\end{definition}

\section{Approximating the spectrum}
\label{s:spectrum}
Recall that we use the following notation throughout the paper. Let $O\colon \iH\to \iH$ be a self-adjoint, trace-class linear operator acting on a separable, complex Hilbert space $\iH$. Then $O$ has only countably many real eigenvalues, let us enumerate them with multiplicities as $\{\lambda_i\}_{i\geq 0}$. As $O$ is trace-class, we have $\sum_{i\geq 0} |\lambda_i|<\infty$. If there are only finitely many eigenvalues $\lambda_0,\dots,\lambda_m$ then let $\lambda_i=0$ for each integer $i>m$.

\subsection{The quantities \texorpdfstring{$e_k$}{TEXT} and their calculation} \label{ss:ek}
Following \cite{Newton} let $e_0=1$ and for positive integers $k\geq 1$ define 
\begin{equation*}
e_k\defeq \sum_{0\leq i_1<\dots<i_k} \lambda_{i_1}\cdots \lambda_{i_k}.
\end{equation*}
It might be standard that the quantities $e_k$ are intimately connected to the positivity of the operator $O$, see \cite[Proposition~2.1]{Newton} for an exact reference. 
\begin{claim} The operator $O$ is positive semidefinite $\Longleftrightarrow $ $e_k\geq 0$ for all $k\geq 1$.
\end{claim}
Let $K\in L^2(\RR^{2d})$ be a self-adjoint kernel. By the Plemelj--Smithies formulas\footnote{See also Newton's identities or the Faddeev–LeVerrier algorithm.} we obtain that the values $e_k$ corresponding to the operator $\widehat{K}$ can be calculated by   
\begin{equation} 
e_k =\frac{1}{k!}\left|\begin{array}{ccccc}
M_1  & 1   & 0  & \cdots       \\
M_2  & M_1 & 2  & 0  & \cdots \\
\vdots  &  & \ddots & \ddots   \\
M_{k-1} & M_{k-2} & \cdots & M_1    & k-1 \\
M_k     & M_{k-1} & \cdots & M_2    & M_1
\end{array}\right|,
\label{eq: e_k-k}
\end{equation}
where the moments $\{M_\ell\}_{1\leq \ell\leq k}$ are defined as 
\begin{equation*}
M_\ell \defeq \sum_{i=0}^{\infty}\lambda_i^\ell=\Tr\big\{\widehat{K}^\ell\big\}
=\int_{\mathbb{R}^{\ell d}} \, \left[K(\mathbf{x}_\ell,\mathbf{x}_1) \prod_{i=1}^{\ell-1} K(\mathbf{x}_i,\mathbf{x}_{i+1}) \right]\, \mathrm{d}\mathbf{x}_1\cdots \mathrm{d}\mathbf{x}_\ell.
\label{eq: M_l-ek}
\end{equation*}
Also, it was demonstrated in \cite{Newton} that this calculation is effective for polynomial Gaussian kernels according to Definition~\ref{d:poly}.

Define $g\colon \RR\to \RR$ as 
 \begin{equation*} 
 g(x)\defeq \prod_{i=0}^{\infty} (1+\lambda_i x).
 \end{equation*}
By \cite[Lemma~3.3]{simon} we obtain that $g(x)$ is finite for all $x\in \RR$ and
 \begin{equation} \label{eq:g} g(x)=\sum_{k=0}^{\infty} e_k x^k,
 \end{equation}
 see also \cite{Newton} for an elementary proof. Taking the algebraic expansion of $\left(\sum_{i=0}^{\infty} |\lambda_i|\right)^k$ yields that the sequence $e_k$ rapidly converges to $0$. More precisely, we obtain the following \cite[Lemma~3.3\,(3.4)]{simon}:
 \begin{equation} \label{eq:ek}
 |e_k|\leq \frac{\left(\sum_{i=0}^{\infty}|\lambda_i| \right)^k}{k!}  
\quad \textrm{for all } k\geq 1.
 \end{equation}

\begin{definition}
\label{def:Lambda_n}
For each integer $n\geq 1$ define the set $\Lambda_n\subset \RR$ as
\begin{equation*} 
\Lambda_n= \left\{-\frac{1}{x}: x\in \RR \text{ and } g_n(x)=0\right\}, \text{ where } g_n(x)=\sum_{k=0}^n e_kx^k.
\end{equation*} 
\end{definition}

The main idea of this section is that $\Lambda_n$ will approximate the closure of the set of non-zero eigenvalues given by
\begin{equation*} 
\left\{-\frac 1x: x\in \RR \text{ and } g(x)=\sum_{k=0}^{\infty} e_k x^k=0\right\}.  
\end{equation*}
The connection between the two sets is that the polynomial $g_n$ is a truncation of the power series $g$. The more precise statement is based on the concept of Hausdorff metric, see Subsection~\ref{ss:theor} for more details.

\subsection{A Gaussian example} \label{ss:EGauss}
First we apply our method for a Gaussian operator, where our approximation $\Lambda_n$ can be directly compared to the exact spectrum. A self-adjoint Gaussian kernel $K_G\colon \RR^2 \to \CC$ with trace $1$ is of the form 
\begin{align}
\begin{split}
K_G(x,y)=N_0\exp&\left\{-A(x-y)^2-iB(x^2-y^2)-C(x+y)^2 \right.\\
&~ \left. -iD(x-y)-E(x+y)\right\},
 	\label{eq:rho_G}
\end{split}
\end{align}
with real parameters $A>0$, $C>0$, $B$, $D$, $E$, where $N_0=2\sqrt{\frac{C}{\pi}}
\exp \left[-\frac{E^2}{4C}\right]$ is a normalizing factor ensuring that $\Tr\big\{\widehat{K}_G\big\}$=1. Let $\beta=\frac{\sqrt{A}-\sqrt{C}}{\sqrt{A}+\sqrt{C}}$, then the eigenvalues are $\lambda_n=(1-\beta)\beta^n$ where $n\geq 0$, and the quantities $e_k$ can be calculated as
\begin{equation*}
e_k=\frac{(1-\beta)^k \beta^{\frac{k(k-1)}{2}}}{\prod_{i=1}^k (1-\beta^i)}.
\label{eq:gaussian_e_k}
\end{equation*}
In Table~\ref{t:tab_gauss} we considered the parameters $A=1$, $C=4$, $B=D=E=0$. 

\begin{table} 
\begin{center}
\begin{tabular}{ |c|c|c|c|}
\hline 
 $k$   & $\lambda_{k}$ & $\lambda^{(10)}_{k}$ & $r_{k} $\\
\hline
$0$ &  $1.333333333$ & $1.333333333$ & $3.86624\cdot 10^{-27}$ \\
\hline
$1$  &   $-0.4444444444$ & $-0.4444444444$ & $1.7122\cdot 10^{-22}$  \\
\hline
$2$  & $0.1481481481$ &  $0.1481481481$ & $3.79164\cdot 10^{-18}$ \\
\hline
$3$ & $ - 0.04938271605 $ & $ - 0.04938271605$ & $2.39836\cdot 10^{-14}$ \\
\hline
$4$ & $0.01646090535$ & $0.01646090536$ & $5.31401\cdot 10^{-11}$ \\
\hline
$5$ &  $ - 0.00548696845$ & $- 0.00548696824$ & $3.85098 \cdot10^{-8}$ \\
\hline
$6$ & $0.001828989483$ &$0.00182897225$ & $9.42254\cdot 10^{-6}$\\
\hline
$7$ & $ - 0.00060966316$  & $ - 0.00061011957$ & $0.000748632$ \\
\hline
$8$ & $0.00020322$ & $0.0002074$ & $0.0204092$ \\
\hline
$9$ & $ - 0.00006774$ & $ - 0.00005448$ & $0.195762$  \\
\hline
\end{tabular} 
\end{center}
\caption{The column $\lambda_k$ contains the first $10$ eigenvalues ordered from the largest to the smallest absolute value. The column $\lambda^{(10)}_k$ contains the elements of the approximation $\Lambda_{10}$ ordered similarly. The last column lists the relative errors $r_{k}=\frac{|\lambda_{k}-\lambda^{(10)}_{k}|}{|\lambda_{k}|}$. Note that, for example, the difference of the real $1$-norm and the one calculated from $\Lambda_{10}$ is just $\sum^{\infty}_{k=0}|\lambda_{k}|-\sum^{9}_{k=0}|\lambda^{(10)}_{k}|=0.000042545$.} 
\label{t:tab_gauss}
\end{table}

\subsection{A non-Gaussian example} \label{ss:NGauss}

 Now we test our estimate in case of an integral operator whose kernel is a quadratic polynomial multiplied by a Gaussian. We will show that 
determining the approximate eigenvalues for positive semidefinite kernels (i.e., for density operators) according to Definition~\ref{def:Lambda_n} allows us to accurately estimate the von Neumann entropy 
\begin{equation*}
S= - \sum^{\infty}_{i=0}\lambda_i \log{\lambda_i} 
\end{equation*}
by
\begin{equation*}
S^{(n)}= - \sum_{i}\lambda^{(n)}_i \log{\lambda^{(n)}_i},
\end{equation*}
where $\lambda^{(n)}_i$ denotes the $i$th eigenvalue in the $n$th approximation. Note that the polynomial $g_n(x)$ of order $n$ in Definition~\ref{def:Lambda_n} is the truncated version of the power series $g(x)$ from
Eq.~\eqref{eq:g}.
Here we set the Boltzmann constant to be $k_B=1$. For non-positive kernels we will also calculate the approximate  Schatten--von Neumann norm by
\begin{equation*}
\| \hat{\rho} \|_1 \approx \sum_{i} |\lambda^{(n)}_i|. 
\end{equation*}

Recall that a polynomial Gaussian kernel $K\colon \RR^2\to \CC$ is given by
\begin{equation*} 
K(x,y)=P(x,y) K_G(x,y),
\end{equation*}
where $K_G(x,y)$ is a self-adjoint Gaussian kernel with trace $1$ of the form \eqref{eq:rho_G} with real parameters $A,B,C,D,E$ such that $A,C>0$; in our example $P(x,y)$ is a self-adjoint quadratic polynomial given by 
\begin{align*} 
P(x,y)=\frac 1N &\left(A_P (x-y)^2+i B_P(x^2-y^2) + C_P (x+y)^2 \right. \\
&\, \left. +i D_P(x-y)+E_P(x+y) +F_P\right), 
\end{align*} 
with real parameters $A_P,B_P, C_P, D_P,E_P,F_P$ and normalizing factor
\begin{equation*} 
N=F_P + \frac{C_P-E_P E}{2C}+\frac{C_P E^2}{4 C^2}.
\end{equation*}

Tables~\ref{t:tab_polynomial_gauss_for_etropy} and \ref{t:tab_polynomial_gauss} show data for a positive semidefinite kernel, and a non-positive kernel, respectively. 
Tables~\ref{t:tab_polynomial_gauss_for_etropy} and \ref{t:tab_polynomial_gauss} indicate data for the von Neumann entropy and the Schatten--von Neumann norm for parameters shown in the captions. One can also compare these values with the more precise values obtained from direct diagonalization (See subsection \ref{ss:num}). The numbers of diagonalization in the captions of Tables~\ref{t:tab_polynomial_gauss_for_etropy} and \ref{t:tab_polynomial_gauss} clearly confirm correctness of the method and the fast convergence to the asymptotics. This fact is a very useful property in practice. These numerical results explicitly demonstrate that our method can accurately approximate key quantum informational quantities such as the von Neumann entropy and logarithmic negativity, even in non-Gaussian, physically motivated scenarios. This confirms the practical relevance of our approach in real-world quantum technologies, especially in continuous variable and noisy quantum systems. The quantities $e_k$ were calculated with \emph{Mathematica}.
%according to Subsection~\ref{ss:ek}. 
%Table~\ref{t:tab_polynomial_gauss} enumerates the values $\lambda^{(n)}_k$ for $n\in \{1,\dots,7\}$  and gives the Schatten--von Neumann norms  for $n=1,\dots,12$.  

%We considered the following example.  Choose the Gaussian parameters as $A=\frac 32$, $C=1$, $B=D=E=0$, and let the polynomial parameters be $A_{P}=-1$, $C_{P}=5$, $F_{P}=1$, and $B_{P}=D_{P}=E_{P}=0$. The quantities $e_k$ were calculated with \emph{Mathematica} according to Subsection~\ref{ss:ek}. Table~\ref{t:tab_polynomial_gauss} summarizes the values $\lambda^{(n)}_k$ for $n,k\in \{1,\dots,6\}$  and the Schatten--von Neumann norm  denotes the approximated eigenvalues in $\Lambda_{n}$ for  $n\in \{1,\dots,12\}$.  

\begin{table}
\[
\begin{array}{|c||c|c|c|c|c|c|}
	\hline
	n & \lambda_1 & \lambda_2 & \lambda_3 & \lambda_4 & \lambda_5 & \lambda_6 \\ \hline
1 & 1.0000000 & & & & & \\
2 & - & & & & & \\
3 & 0.6270753 & 0.3061496 & 0.0667750 & & & \\
4 & 0.6261790 & 0.3103465 & 0.0538776 & 0.0095969 & & \\
5 & 0.6261805 & 0.3103324 & 0.0541775 & 0.0080588 & 0.0012508 & \\
6 & 0.6261805 & 0.3103324 & 0.0541768 & 0.0080884 & 0.0010669 & 0.0001550\\
7 & 0.6261805 & 0.3103324 & 0.0541768 & 0.0080883 & 0.0010701 & 0.0001333 \\
\hline
\end{array}
\]
\\

\[
\begin{array}{|c||c|c|c|c|c|c|}
\hline
	n & 1 & 2 & 3 & 4 & 5 & 6 \\
	\hline
    S^{(n)} & 0.0000000 & - & 0.8357542 & 0.8582210 & 0.8614143 & 0.8618263 \\
\hline\hline
	n & 7 & 8 & 9 & 10 & 11 & 12 \\
	\hline  
    S^{(n)} & 0.8618767 & 0.8618827 & 0.8618834 & 0.8618835 & 0.8618835 & 0.8618835 \\
    \hline
\end{array}
\]
\caption{Eigenvalues (upper table) and von Neumann's entropy (lower table) obtained from different order of approximation $n$ for the polynomial Gaussian kernel $K(\mathbf{x},\mathbf{y})$. We chose the Gaussian parameters as $A=\frac 32$, $C=1$, $B=D=E=0$, and let the polynomial parameters be $A_{P}=-1$, $C_{P}=1$, $F_{P}=1$, and $B_{P}=D_{P}=E_{P}=0$. The von Neumann's entropy from diagonalization is: $S=0.86188346992$.}
\label{t:tab_polynomial_gauss_for_etropy}
\end{table}

\begin{table}
\[
\begin{array}{|c||c|c|c|c|c|c|}
	\hline
	n & \lambda_1 & \lambda_2 & \lambda_3 & \lambda_4 & \lambda_5 & \lambda_6 \\ \hline
1 & 1.0000000 & & & & &  \\
2 & 0.5253423 & 0.4746577 & & & &  \\
3 & 0.5747363 & 0.4133032 & 0.0119604 & & & \\
4 & 0.5746924 & 0.4133888 & 0.0110457 & 0.0008731 & & \\
5 & 0.5746912 & 0.4133922 & 0.0085830 & 0.0062595 & -0.0029260 & \\
6 & 0.5746912 & 0.4133922 & 0.0086980 & 0.0060556 & -0.0030094 & 0.0001724 \\
7 & 0.5746912 & 0.4133922 & 0.0086978 & 0.0060563 & -0.0030100 & 0.0001458 \\
\hline
%	  1 & 1.00000 & & & & &  \\	
%	2 & 0.525217 & 0.474783 & & & &  \\
%	3 & 0.574737  &  0.413287 & 0.0119764 &   &   &  \\	
%	4 & 0.574692  & 0.413374  & 0.0110462  & 0.00088780 &   &    \\
%	5 & 0.574691  & 0.413377  & 0.0085829  & 0.00626675  & -0.00291773  &      \\
%	6 & 0.574691  & 0.413377  & 0.0086981  &  0.00606272 & -0.00300139  & 0.000172445     \\
%	7 & 0.574691  &  0.413377 &  0.0086978 & 0.00606346  &  -0.00300199 & 0.000145851  & 0.0000267351  \\	
%	\hline
\end{array}
\]
\\

\[
\begin{array}{|c||c|c|c|c|c|c|}
	\hline
	n & 1 & 2 & 3 & 4 & 5 & 6 \\
	\hline
\|\hat{\rho}\|_1 & 1.0000000 & 1.0000000 & 1.0000000 & 1.0000000 & 1.0058519 & 1.0060189 \\
%	\|\hat{\rho}\|_1 & 1.000000 & 1.000000 & 1.000000 & 1.000000 & 1.0058355 & 1.0060028 \\
%	0.605432 & 0.837792 & 1. & 1. & 1. & 1. \\
	\hline\hline
	n & 7 & 8 & 9 & 10 & 11 & 12 \\
	\hline
\|\hat{\rho}\|_1 & 1.0060201 & 1.0060201 & 1.0060201 & 1.0060201 & 1.0060201 & 1.0060201 \\
%	\|\hat{\rho}\|_1 & 1.006004 & 1.006004 & 1.006004  & 1.006004  & 1.006004  &  1.006004  \\ 
%	&  1. & 1. & 1. & 1. & 1. & 1.\\
	\hline
\end{array}
\]
\caption{Eigenvalues (upper table) and Schatten--von Neumann norm (lower table) obtained from different order of approximation $n$ for the polynomial Gaussian kernel $K(\mathbf{x},\mathbf{y})$. We chose the Gaussian parameters as $A=\frac 32$, $C=1$, $B=D=E=0$, and let the polynomial parameters be $A_{P}=-1$, $C_{P}=5$, $F_{P}=1$, and $B_{P}=D_{P}=E_{P}=0$. The Schatten--von Neumann norm and the only negative eigenvalue $\lambda_5$ from diagonalization are: $\|\hat{\rho}\|_1=1.00602007614$, $\lambda_5=-0.00301003807$ .}
\label{t:tab_polynomial_gauss}
\end{table}

\subsection{A theoretical result} \label{ss:theor}
The main goal of this subsection is to show that the set $\Lambda_n$ defined in Definition~\ref{def:Lambda_n} converges to the closure of the set of non-zero eigenvalues. Indeed, as the quantities $e_k$ from Subsection~\ref{ss:ek} carry no information about the zero eigenvalues, we need to implement a mild modification to the set of eigenvalues. Let $\Lambda\subset \RR$ denote the closure of the set of non-zero $\lambda_i$, that is,  
\begin{equation*}
\Lambda=\overline{\{\lambda_i: i\geq 0\}\setminus\{0\}}.
\end{equation*}
This means that $0\in \Lambda$ if and only if there are infinitely many non-zero $\lambda_i$.

\begin{definition} \label{d:Haus} 
For two non-empty compact sets $K_1,K_2\subset \RR^d$ let us define their \emph{Hausdorff distance} as 
\begin{equation*}
\dH(K_1,K_2)\defeq \inf\{r: K_1\subset B(K_2,r) \text{ and } K_2\subset B(K_1,r) \},  
\end{equation*}
where 
\begin{equation*} B(A,r)=\{x\in \RR^d: \exists y\in A \text{ such that } |x-y|<r\}.
\end{equation*} 
Then the non-empty compact subsets of $\RR^d$ endowed with the \emph{Hausdorff metric} $\dH$ form a complete, separable metric space, see \cite{kechris} for more information. 
\end{definition}

The main goal of this subsection is to prove the following theorem.  We may assume that $O$ is not the null operator, so $\Lambda\neq \emptyset$.

\begin{theorem} \label{t:conv}
Assume that each non-zero $\lambda_i$ has multiplicity $1$. Then $\Lambda_n\to \Lambda$ in Hausdorff metric.\footnote{This means that, in particular, for large $n$ the set $\Lambda_n$ is non-empty.} 
\end{theorem}
\begin{proof}
First assume that there are only finitely many non-zero $\lambda_i$, say $\lambda_0, \dots, \lambda_{k-1}$. Then clearly $g_n(x)=g(x)$ for all $n\geq k$ for each $x$. Hence $\Lambda_n=\Lambda$ for all $n\geq k$.  

Now assume that there are infinitely many non-zero $\lambda_i$. Fix an arbitrary $\varepsilon>0$, we need to prove that there exists $N\in \NN$ such that $\dH(\Lambda_n, \Lambda)\leq \varepsilon$ for all $n\geq N$. As $0\in \Lambda$, we can define 
\begin{equation*} 
K_{\varepsilon}=\left\{-\frac 1x: x\in \RR\setminus B(\Lambda,\varepsilon)\right\}\cup \{0\}.
\end{equation*}
Clearly $K_{\varepsilon}$ is closed, and $K_{\varepsilon}\subset [-1/\varepsilon, 1/\varepsilon]$, so $K_{\varepsilon}$ is compact. As $g$ is continuous and non-zero on $K_{\varepsilon}$, we can define $\theta\defeq \min\{|g(x)|: x\in K_{\varepsilon}\}>0$. Since $g_n\to g$ uniformly on $K_{\varepsilon}$, there exists $N_1\in \NN$ such that $|g_n(x)-g(x)|<\theta$ for all $n\geq N_1$ and $x\in K_{\varepsilon}$. Therefore, the definition of $\theta$ yields $K_{\varepsilon}\cap \{x: g_n(x)=0\}=\emptyset$ for all $n\geq N_1$. Hence the definitions of $K_{\varepsilon}$ and $\Lambda_n$ imply for each $n\geq N_1$ that
\begin{equation} \label{eq:cont1}
\Lambda_n\subset B(\Lambda,\varepsilon).
\end{equation}
Now choose finitely many non-zero $\lambda_i$ (we may assume that $\lambda_0,\dots, \lambda_{k-1}$) and $\delta>0$ such that 
\begin{equation} \label{eq:ldelta} 
[\lambda_i-\delta, \lambda_i+\delta]\cap \Lambda=\{\lambda_i\} \quad \text{for all } 0\leq i<k,
\end{equation} 
and for any points $x_i\in [\lambda_i-\delta, \lambda_i+\delta]$ we have 
\begin{equation} \label{eq:ldelta2} 
\Lambda\subset B\left(\cup_{i=0}^{k-1} \{x_i\},\varepsilon\right).
\end{equation} 
Define a $2k$ element set 
\begin{equation*} 
C\defeq \left\{\frac{1}{-\lambda_i\pm \delta}: 0\leq i<k\right\},
\end{equation*} 
and let $I_i$ be the open interval with endpoints $1/(-\lambda_i\pm\delta)$ for any $0\leq i<k$. As $\lambda_i$ has multiplicity $1$, \eqref{eq:ldelta} yields that $g(\frac{1}{-\lambda_i-\delta})$ and $g(\frac{1}{-\lambda_i+\delta})$ are non-zero and have opposite signs for all $i$. As $g_n\to g$ uniformly on $C$, we infer that there exists $N_2\in \NN$ such that for each $n\geq N_2$ the values $g_n(\frac{1}{-\lambda_i-\delta})$ and $g_n(\frac{1}{-\lambda_i+\delta})$ are non-zero and have opposite signs for all $0\leq i<k$. Thus for any $n\geq N_2$ for all $0\leq i<k$ by the continuity of $g_n$ we can define $z_{i,n}\in I_i$ such that $g_n(z_{i,n})=0$. Let $x_{i,n}=-1/z_{i,n}$, then clearly $\{x_{i,n}: 0\leq i<k\}\subset \Lambda_n$, and $x_{i,n}\in (\lambda_i-\delta, \lambda_i+\delta)$ for all $0\leq i<k$. Then \eqref{eq:ldelta2} implies that 
\begin{equation} \label{eq:cont2}
\Lambda\subset B(\{x_{i,n}: 0\leq i<k\},\varepsilon) \subset B(\Lambda_n,\varepsilon).
\end{equation} 
Finally, let $N=\max\{N_1,N_2\}$, then \eqref{eq:cont1} and \eqref{eq:cont2} imply for all $n\geq N$ that
\begin{equation*}
\dH(\Lambda_n,\Lambda)\leq \varepsilon. 
\end{equation*} 
As $\varepsilon>0$ was arbitrary, the proof is complete. 
\end{proof}

\begin{remark}
The multiplicity assumption in the above theorem cannot be omitted entirely, as eigenvalues of even multiplicity can cause problems. However, this assumption can be removed by allowing complex roots of the approximating polynomials $g_n$; see Appendix~\ref{app:A}.
\end{remark}

\section{A non-monotone approximation of the minimal eigenvalue} \label{s:qn0}

We apply the method of Section~\ref{s:spectrum} to estimate the minimal eigenvalue. Let the operator $O$ and its real eigenvalues $\lambda_i$ be the same as in Section~\ref{s:spectrum}. Recall that $\lambda_{\min}=\min\{0,\lambda_i: i\geq 0\}$, and note that
 \begin{equation} \label{eq:lmin} \lambda_{\min}=\frac{-1}{\min\{x>0: g(x)=0\}}, 
 \end{equation} 
 where we use the convention $\min \emptyset=+\infty$. Let us recall the following definition. 

\begin{definition}
For all integers $n\geq 1$ define 
\begin{equation} \label{eq:qn0}
q_{n,0}\defeq \frac{-1}{\min \left\{x>0: g_n(x)=0\right\}}, \quad \text{where } g_n(x)=\sum_{k=0}^n e_kx^k.
\end{equation} 
\end{definition}

We consider the approximation of $\lambda_{\min}$ by $q_{n,0}$. In the next theorem we prove that $q_{n,0}\to \lambda_{\min}$ with super-exponential speed under the slight technical condition that $\lambda_{\min}<0$ has multiplicity one. First we need the following easy fact. 

\begin{fact} \label{f:tail} 
Let $n\geq 0$ be an integer and 
let $0\leq x\leq \frac{n+2}{2}$. Then 
\begin{equation*}
\sum_{k=n+1}^{\infty} \frac{x^k}{k!}\leq \frac{2x^{n+1}}{(n+1)!}
\end{equation*} 
\end{fact} 

\begin{proof}
We can use a geometric sum for the estimate as 
\begin{equation*} 
\sum_{k=n+1}^{\infty}  \frac{x^k}{k!}\leq \frac{x^{n+1}}{(n+1)!}\sum_{i=0}^{\infty} \left(\frac{x}{n+2}\right)^i\leq \frac{2x^{n+1}}{(n+1)!}. \qedhere
\end{equation*}
\end{proof}

\begin{theorem} \label{t:c=0}
Assume that $\lambda_{\min}<0$ has multiplicity one. Then $q_{n,0}\to \lambda_{\min}$. Moreover, 
\begin{equation} \label{eq:pn}
|q_{n,0}-\lambda_{\min}|\leq \exp[-n \log(n)+\mathcal{O}(n)]. 
\end{equation}
\end{theorem} 

\begin{proof} 
First we prove that $q_{n,0}\to \lambda_{\min}$. Assume that $\lambda_0=\lambda_{\min}$ and fix $\gamma_0$ such that $\lambda_0<\gamma_0<\min\{0, \lambda_i: i\geq 1\}$. Define 
\begin{equation*} 
\theta_0=-\frac{1}{\lambda_0} \quad \text{and} \quad \theta_1=-\frac{1}{\gamma_0}.
\end{equation*}
Since $g$ has no root between $\theta_0$ and $\theta_1$, and $\theta_0$ is the smallest root of $g$ with multiplicity $1$, we obtain that $g(x)>0$ for all $0\leq x<\theta_0$ and $g(x)<0$ for all $\theta_0<x\leq \theta_1$. Fix an arbitrary  $0<\varepsilon<\min\{\theta_0,\theta_1-\theta_0\}$, and define $\delta>0$ as 
\begin{equation*} 
\delta\defeq \min\{|g(x)|: x\in [0, \theta_0-\varepsilon]\cup\{\theta_0+\varepsilon \}  \}. 
\end{equation*}  
Since $g_n$ uniformly converges to $g$ on the interval $[0, \theta_0+\varepsilon]$, we obtain that there exists an integer $N=N(\varepsilon)\geq 0$ such that 
\begin{equation} \label{eq:hgdelta} |g_n(x)-g(x)|<\delta \quad \text {for all } x\in  [0, \theta_0+\varepsilon] \text{ and } n\geq N.
\end{equation}
The definition of $\delta$ and \eqref{eq:hgdelta} imply that for all $n\geq N$ we have $g_n(x)>0$ for all $x\in [0, \theta_0-\varepsilon]$ and $g_n(\theta_0+\varepsilon)<0$. Hence $\min \left\{x>0: g_n(x)=0\right\}\in [\theta_0-\varepsilon, \theta_0+\varepsilon]$ for all $n\geq N$. Taking the limit $\varepsilon \to 0+$ yields that $ \min \left\{x>0: g_n(x)=0\right\}\to \theta_0$ as $n\to \infty$, which implies that $q_{n,0}\to \lambda_0$ by definition. 

Now we prove the upper bound \eqref{eq:pn}. Let $F(x)=|\lambda_0| \prod_{i\geq 1} (1+\lambda_i x)$, and let 
\begin{equation*}
\Delta\defeq \min\left\{F(x): 0\leq x\leq \theta_1\right\}.
\end{equation*}
As $F$ is continuous and positive on the interval $[0,\theta_1]$, we obtain that $\Delta>0$. Consider $x=\theta_0+s$, where $s\in [-\theta_0,\theta_1-\theta_0]$. Then the definition of $\Delta$ yields
\begin{equation}
\label{eq:glin}
|g(x)|=\left|\prod_{i=0}^{\infty} (1+\lambda_ix)\right|=|s \lambda_0|  \prod_{i\geq 1} (1+\lambda_i x)\geq  \Delta|s|.
\end{equation}
Define 
\begin{equation*} 
L\defeq \sum_{i=0}^{\infty} |\lambda_i|,
\end{equation*} 
and assume that $n>2 \theta_1 L$. Then $Lx\leq L\theta_1<\frac n2$ for every $x\in [0,\theta_1]$. Thus \eqref{eq:ek}, Fact~\ref{f:tail}, and the inequality $n!\geq e\left(\frac{n}{e}\right)^{n}$ imply that for all $0\leq x\leq \theta_1$ we have  
\begin{equation} \label{eq:hg}  |g_n(x)-g(x)|=\left|\sum_{k=n+1}^{\infty} e_kx^k\right|\leq \sum_{k=n+1}^{\infty} \frac{(Lx)^k}{k!}\leq 2\frac{(Lx)^{n+1}}{(n+1)!}< \left(\frac{eL\theta_1}{n+1}\right)^{n+1}. 
\end{equation}
Now define $\varepsilon_n>0$ as 
\begin{equation} \label{eq:eps} 
\varepsilon_n \defeq \frac{1}{\Delta}\left(\frac{eL\theta_1}{n+1}\right)^{n+1}. 
\end{equation} 
Then \eqref{eq:glin}, \eqref{eq:hg}, and \eqref{eq:eps} yield that for all large enough $n$ we have $g_n(x)>0$ for all $0\leq x\leq \theta_0-\varepsilon_n$ and $g_n(\theta_0+\varepsilon_n)<0$; 
so we can define 
\begin{equation*} r_n\defeq \min\{x>0: g_n(x)=0 \}\in [\theta_0-\varepsilon_n, \theta_0+\varepsilon_n].
\end{equation*} 
Then $r_n\in [\theta_0-\varepsilon_n, \theta_0+\varepsilon_n]$, and the definition of $\varepsilon_n$ with some straightforward analysis imply that for all large enough $n$ we have
\begin{equation*} 
|q_{n,0}-\lambda_0|=\left|\frac{1}{\theta_0}-\frac{1}{r_n}\right|\leq \left|\frac{1}{\theta_0}-\frac{1}{\theta_0-\varepsilon_n}\right|\leq 2|\lambda_0|^2 |\varepsilon_n|=\exp[-n \log(n)+\mathcal{O}(n)].
\end{equation*} 
The proof of the theorem is complete.
\end{proof} 

Figure~\ref{fig:nongaussian_qns} will compare this approximation with the one in Section~\ref{s:lmin}. 

\begin{remark}
Note that the proof of an analogous theorem would also work if we only assumed that $\lambda_{\min}<0$ has \emph{odd} multiplicity $m$ in the sequence $\{\lambda_i\}_{i\geq 0}$, with the only difference that $n\log(n)$ need to be replaced by $\frac 1m n\log(n)$ in \eqref{eq:pn}. However, having multiplicity bigger than one does not seem to be a natural condition for the minimal eigenvalue of a linear operator. 
\end{remark}

\section{Approximation of the Schatten--von Neumann 1-norm} \label{s:Sch} 
The main goal of this section is to give a method which can provide rigorous upper bounds for the Schatten--von Neumann $1$-norm, and to demonstrate that it works effectively for polynomial Gaussian operators, see Definition~\ref{d:poly}. We use this  section as a preparation for Section~\ref{s:lmin}, but it can be interesting in its own right. Let $\iH$ be a separable, complex Hilbert space. 
For the following definition see for example \cite[Chapter~VI]{SR}.
\begin{definition} Assume that $O\colon \iH\to \iH$ is a Hilbert--Schmidt operator. Let $\{s_i\}_{i\geq 0}$ be the \emph{singular values} of $O$, that is, the sequence of the square roots of the eigenvalues of $OO^{\dagger}$ with multiplicities. Then we can define the finite $2$-norm of $O$ as 
\begin{equation*}
\|  O \|_2 \defeq \Bigg( \sum_{i\geq 0} s_{i}^2 \Bigg)^{\frac{1}{2}}. 
\end{equation*}
If $O$ is trace-class, we can define its \emph{Schatten--von Neumann $1$-norm} as
\begin{equation*}
\|  O \|_1 \defeq \sum_{i\geq 0} s_{i}. 
\end{equation*}
\end{definition}
By \cite[Theorem~VI.23]{SR} for $K\in L^2(\RR^{2d})$ the $2$-norm of the Hilbert--Schmidt integral operator $\widehat{K}$ can be calculated as 
\begin{equation} \label{eq:Knorm}
\big\|  \widehat{K}\big\|^2_2=\int_{\mathbb{R}^{2d}}  |K(\mathbf{x},\mathbf{y}) |^2 \, \mathrm{d} \mathbf{x} \, \mathrm{d} \mathbf{y}.
\end{equation} 

Now assume that $O$ is a self-adjoint, trace-class operator, then clearly we obtain that $\|  O \|_1=\sum_{i\geq 0} |\lambda_i|<\infty$, where $\{\lambda_i\}_{i\geq 0}$ are the eigenvalues with multiplicities. The main goal of this section is to find an upper bound for $\|O\|_1$ which can be effectively calculated. By \cite[Theorem~VI.~22]{SR} the operator $O$ can be written as a product 
\begin{equation} \label{eq:O12}
O=O_1 O_2,
\end{equation}
where $O_1,O_2\colon \iH \to \iH$ are Hilbert--Schmidt operators, so $\|O_i\|_2<\infty$ for $i=1,2$. H\"older's inequality \cite[Problem~28 in Chapter~VI]{SR} yields
\begin{equation} \label{eq:OHolder}
\|O\|_1\leq \| O_1\|_2 \| O_2\|_2, 
\end{equation} 
giving an upper bound for the Schatten--von Neumann $1$-norm of $O$. However, it is a hard problem to find a decomposition \eqref{eq:O12} where we can also calculate or estimate the norms $\| O_1\|_2$ and $\| O_2\|_2$. 

We propose a solution to self-adjoint trace-class integral operators $\widehat{K}$ with kernels consisting of a polynomial multiplied by a Gaussian kernel, see Subsection~\ref{ss:DO}. 
Then a decomposition to Hilbert--Schmidt operators
\begin{equation} \label{eq:K12}
\widehat{K}=\widehat{K}_1 \widehat{K}_2
\end{equation}
implies that the kernels satisfy
\begin{equation} \label{eq:prod_of_int_ops}
K(\mathbf{x},\mathbf{y}) = \int_{\mathbb{R}^d}  K_1(\mathbf{x},\mathbf{z}) K_2(\mathbf{z},\mathbf{y}) \, \mathrm{d} \mathbf{z}.
\end{equation}
Applying H\"older's inequality \eqref{eq:OHolder} to the decomposition \eqref{eq:K12}, and using the formula \eqref{eq:Knorm} for $K_1$ and $K_2$ yield
\begin{equation} \label{eq:Hold}
\big\|  \widehat{K} \big\|_1\leq \left(\int_{\mathbb{R}^{2d}}  |K_1(\mathbf{x},\mathbf{y}) |^2 \, \mathrm{d} \mathbf{x} \, \mathrm{d} \mathbf{y}\right)^{\frac{1}{2}} \left(\int_{\mathbb{R}^{2d}}  |K_2(\mathbf{x},\mathbf{y}) |^2 \, \mathrm{d} \mathbf{x} \, \mathrm{d} \mathbf{y}\right)^{\frac{1}{2}}. 
\end{equation}
This upper bound has a direct connection to the minimal eigenvalue problem. We obtain an immediate lower bound for $\lambda_{\min}$ and upper bound for the negativity $\mathcal{N} \defeq \sum_{i=0}^{\infty} |\min\{\lambda_i,0\}|$ as 
\begin{equation} \label{eq:Trmin}
|\lambda_{\min}|\leq \mathcal{N}=
\frac{\big\|  \widehat{K} \big\|_1- \Tr \big\{\widehat{K}\big\}}{2}\leq \frac{\big\|  \widehat{K}_1  \big\|_{2} \big\|  \widehat{K}_2  \big\|_{2}- \Tr \big\{\widehat{K}\big\}}{2}.
\end{equation}
It is interesting that while our operator $\widehat{K}$ is always self-adjoint, often there are only decompositions of the form \eqref{eq:K12} in which neither $\widehat{K}_1$ nor $\widehat{K}_2$ is self-adjoint. We will see more sophisticated estimates of $\lambda_{\min}$ in Section~\ref{s:lmin}, where we will also use an upper bound for $\big\|  \widehat{K} \big\|_1$ coming from \eqref{eq:K12} and \eqref{eq:Hold}. In the rest of this section we show a method how to calculate such upper bounds.

\subsection{Decomposition and optimisation} \label{ss:DO}

Given a polynomial Gaussian kernel $K\in L^2(\RR^{2d})$ of the form \eqref{eq:poly}, we want to decompose the operator $\widehat{K}$ to a product according to \eqref{eq:K12}. 
We will search the decomposition in the form 
\begin{equation*} 
K_1(\mathbf{x},\mathbf{y})=P_1(\mathbf{x},\mathbf{y})K_{G_1}(\mathbf{x},\mathbf{y}) \quad \text{and} \quad  K_2(\mathbf{x},\mathbf{y})=P_2(\mathbf{x},\mathbf{y})K_{G_2}(\mathbf{x},\mathbf{y}), 
\end{equation*}
where $K_{G_1},K_{G_2}$ are not necessarily self-adjoint Gaussian kernels in $2d$ variables, and $P_1,P_2$ are polynomials in $2d$ variables such that $\deg P_1+\deg P_2=\deg P$. In particular, if $\deg P_1=0$ or $\deg P_2=0$ then the decomposition leads us to a system of linear equations which can be solved effectively, and we can even try to minimize $\big\|  \widehat{K}_1  \big\|_{2} \big\|  \widehat{K}_2  \big\|_{2}$ coming from the different decompositions of $\widehat{K}$. We demonstrate this algorithm in the $d=1$ case together with a minimization in the following subsections.

\subsection{Gaussian operators in dimension one} \label{ss:Gauss}

Recall that a self-adjoint  Gaussian kernel $K_G\colon \RR^2\to \CC$ is of the form
\begin{align} \label{eq:g1d}
\begin{split}
K_G(x,y)=N_0\exp&\left\{-A(x-y)^2-iB(x^2-y^2)-C(x+y)^2 \right.\\
&~ \left. -iD(x-y)-E(x+y)\right\},
\end{split}
\end{align}
with real parameters $A>0$, $C>0$, $B$, $D$, $E$, where $N_0=2\sqrt{\frac{C}{\pi}}
\exp \left[-\frac{E^2}{4C}\right]$ is a normalizing factor ensuring that $\Tr\big\{\widehat{K}_G\big\}$=1.
The eigenvalue equation was solved in \cite{Bernád2018}; the eigenvalues are given by
\begin{equation*}
\lambda_i=\frac{2\sqrt{C}}{\sqrt{A}+\sqrt{C}} \left(\frac{\sqrt{A}-\sqrt{C}}{\sqrt{A}+\sqrt{C}}\right)^i, \quad i=0,1,\ldots.
\end{equation*} 
Correspondingly, the Schatten--von Neumann $1$-norm of the Gaussian $\widehat{K}_G$ is
\begin{equation} \label{eq:exact_l1norm:for_gauss}
\big\| \widehat{K}_G\big\|_1=\sum_{i=0}^\infty |\lambda_i|=\begin{cases}
	1 & \text{if } A\geq C, \\
	\sqrt{\frac CA} & \text{if } C>A.
\end{cases}
\end{equation}
By \eqref{eq:prod_of_int_ops} we consider the decomposition of $K_G$ as 
\begin{equation} \label{eq:gauss_decomp}
K_G (x,y)=\int_{-\infty}^\infty \, K_1(x,z) K_2(z,y)\, \mathrm{d}z. 
\end{equation}
We define $K_1$ and $K_2$ containing the free parameter $w\in \RR$ as follows. Let
\begin{equation*}
K_1(x,y)=N_1 \exp \left(-A_1 x^2 - B_1 y^2 - C_1 xy - D_1 x - E_1 y \right),
\end{equation*}
with
\begin{align} 
\begin{split}
\label{eq:first_exp_par1}
A_1&=w+ i B +\frac{AC}{w}\\ 
B_1&=w- i B +\frac{AC}{w},\\ 
C_1&=-2 \left( w-\frac{AC}{w}\right), \\ 
D_1&= i D+ \frac{AE}{w}, \\ 
E_1&=-i D+ \frac{AE}{w},  \\ 
N_1&=\frac{2\sqrt{C}}{\pi}\sqrt{\frac{(w^2-AC)^2}{w(A-w)(C-w)}}\exp{\left(-\frac{\left(AC-w^2\right)E^2}{4w(C-w)C}\right)}. 
\end{split}
\end{align}
Similarly let 
\begin{equation*} 
K_2(x,y)= \exp \left(-A_2 x^2 - B_2 y^2 - C_2 xy - D_2 x - E_2 y \right)
\end{equation*}
with
\begin{align}
\begin{split}
\label{eq:second_exp_par1}
A_2&= i B +\frac{(A-w)C}{C-w}+\frac{A(C-w)}{A-w},\\
B_2&= - i B  + \frac{(A-w)C}{C-w}+\frac{A(C-w)}{A-w},  \\
C_2&=2 \left( \frac{(A-w)C}{C-w}-\frac{A(C-w)}{A-w} \right), \\   
D_2&= i D + \frac{(A-w)E}{C-w},  \\  
E_2&=-i D + \frac{(A-w)E}{C-w}. 
\end{split}
\end{align}
It is easy to check that \eqref{eq:gauss_decomp} holds and the integral is finite if $w$ satisfies
\begin{equation} \label{eq:allowed_w}
0 < w < \min(A,C)  \qquad \text{or} \qquad 0< \max(A,C) < w.
\end{equation}
Straightforward calculation leads to
\begin{align}
\begin{split}
\label{eq:rho12}
\iint_{\RR^2} \left| K_1(x,y)\right|^2 \, \mathrm{d} x \, \mathrm{d} y=&\frac{C(w^2-A C)^2}{\pi \sqrt{AC} w(A-w)(C-w)} \exp{\left( -\frac{(A-w)E^2}{2(C-w)C}\right)},\\
\iint_{\RR^2} \left| K_2(x,y)\right|^2 \, \mathrm{d} x \, \mathrm{d} y=&\frac{\pi}{4\sqrt{AC}} \exp{\left( \frac{(A-w)E^2}{2(C-w)C}\right)},
\end{split}
\end{align}
provided that $w$ is taken from the allowed intervals \eqref{eq:allowed_w}. By \eqref{eq:rho12} we can calculate the product
\begin{equation*}
\iint_{\RR^2} \left| K_1(x,y)\right|^2 \, \mathrm{d} x \, \mathrm{d} y \iint_{\RR^2} \left| K_2(x,y)\right|^2 \, \mathrm{d} x \, \mathrm{d} y=\frac{(w^2-AC)^2}{4Aw(A-w)(C-w)}.
\end{equation*}
We can minimize this in $w$ on $[0,\min[A,C]]\cup [\max[A,C],\infty]$, which gives 
\begin{eqnarray} \label{eq:location_of_min_for_pos_gauss}
w_{\min}^{G}=\begin{cases}
A \mp \sqrt{A^2-AC}&  \qquad \text{if } A\geq C, \\
C \mp \sqrt{C^2-AC}& \qquad \text{if } C>A.
\end{cases}
\end{eqnarray}
The minimum value are the same at  both $(\pm)$ values.
Therefore, using H\"older's inequality \eqref{eq:Hold} 
we obtain the upper bound 
\begin{align*}
\big\| \widehat{K}_G\big\|_1&\leq 
\left. \sqrt{\iint_{\RR^2} \left| K_1(x,y)\right|^2 \, \mathrm{d} x \, \mathrm{d} y  \iint_{\RR^2} \left| K_2(x,y)\right|^2 \, \mathrm{d} x \, \mathrm{d} y	}\right|_{w=w_{\min}^{G}} \\
&=\begin{cases}
	1 & \text{if } A\geq C, \\
	\sqrt{\frac{C}{A}}& \text{if } C<A.
\end{cases}
\end{align*}
Comparing this with \eqref{eq:exact_l1norm:for_gauss} shows that our upper bound equals to the exact value of the norm $\big\| \widehat{K}_G\big\|_1$. 

\subsection{Quadratic polynomials multiplied by a Gaussian} \label{ss:quad}

Our second example is the polynomial Gaussian kernel $K\colon \RR^2\to \CC$ given by
\begin{equation} \label{eq:rho_Q}
K(x,y)=P(x,y) K_G(x,y),
\end{equation}
where $K_G(x,y)$ is the self-adjoint Gaussian kernel with trace $1$ from \eqref{eq:g1d}, and $P(x,y)$ is the self-adjoint polynomial given by 
\begin{align*} 
P(x,y)=\frac 1N &\left(A_P (x-y)^2+i B_P(x^2-y^2) + C_P (x+y)^2 \right. \\
&\, \left. +i D_P(x-y)+E_P(x+y) +F_P\right), 
\end{align*} 
with real parameters $A_P,B_P, C_P, D_P, E_P, F_P$ and normalizing factor
\begin{equation*} 
N=F_P + \frac{C_P-E_P E}{2C}+\frac{C_P E^2}{4 C^2}.
\end{equation*}
This ensures that $\Tr\big\{ \widehat{K}\big\}=1$. 
Let us decompose $\widehat{K}=\widehat{K}_1 \widehat{K}_2$, by \eqref{eq:prod_of_int_ops} the kernels must satisfy 
\begin{equation} \label{eq:decomposition_B}
K(x,y)=\int_{-\infty}^{\infty} K_1(x,z)K_2(z,y)\, \mathrm{d}z.
\end{equation}
We assume that $K_1$ is a Gaussian kernel and $K_2$ is a Gaussian multiplied by a quadratic polynomial of the form
\begin{align*}
K_1(x,y)&=N'_1 \exp\left( -A_1 x^2 - B_1 y^2 - C_1 xy - D_1 x - E_1 y\right); \\
K_2(x,y)&=P_2(x,y)\exp\left( -A_2 x^2 - B_2 y^2 - C_2 xy - D_2 x - E_2 y\right),
\end{align*}
where 
\begin{align*} 
P_2(x,y)=&A_{P_2} (x-y)^2+i B_{P_2}(x^2-y^2) + C_{P_2} (x+y)^2\\
&+i D_{P_2}(x-y)+E_{P_2}(x+y) +F_{P_2}, 
\end{align*}
and all the parameters are real. Choose the exponential parameters $A_1,\dots,E_1$, and $A_2,\dots,E_2$ according to \eqref{eq:first_exp_par1}-\eqref{eq:second_exp_par1}. By a straightforward calculation the normalizing factor $N'_1$ and the polynomial parameters $A_{P_2},\dots,F_{P_2}$ are uniquely determined by the decomposition \eqref{eq:decomposition_B}: one obtains linear equations for the polynomial parameters. Then the only free parameter is  $w$ (as before), which allows minimization. Allowed regions for $w$ are still given by \eqref{eq:allowed_w}, where $A$ and $C$ are real parameters of the Gaussian $K_G(x,y)$. The quantity
\begin{equation*}
R(w)=\left(\iint_{\RR^2} \left| K_1(x,y)\right|^2 \, \mathrm{d} x \, \mathrm{d} y  \right) \left( \iint_{\RR^2} \left| K_2(x,y)\right|^2 \, \mathrm{d} x \, \mathrm{d} y\right)
\end{equation*}
can be easily calculated. It is a lengthy rational function of the free parameter $w$, thus we do not quote it here. Finding the minimum point $w_{\min}$ of $R(w)$ in the allowed  regions \eqref{eq:allowed_w} give an optimized upper bound to the norm $\big\| \widehat{K}\big\|_1$. Note that the value of $w_{\min}$ is different from $w_{\min}^{G}$  given in \eqref{eq:location_of_min_for_pos_gauss}, our minimum is affected by the polynomial parameters as well. On Figure~\ref{fig:L1norm_as_a_function_gamma2} we display upper bounds of the Schatten--von Neumann $1$-norm for a family of quadratic Gaussian operators of the form \eqref{eq:rho_Q}. 
The decomposition \eqref{eq:decomposition_B} is not as optimal as in our first example. It can be proved that at $C_{P}=1$ the the operator $\widehat{K}$ is positive semidefinite, thus 
the $1$-norm is exactly $1$, but our decomposition \eqref{eq:decomposition_B} after the minimization leads to the upper bound 
\begin{equation} \label{eq:104}
\big\|\widehat{K} \big\|_1\leq 1.04054.
\end{equation}
Then \eqref{eq:Trmin} yields a lower bound for the minimal eigenvalue $\lambda_{\min}$ and an upper bound for the negativity $\mathcal{N}$ as follows:
\begin{equation} \label{eq:lminbound} 
|\lambda_{\min}|\leq \iN\leq \frac{1.04054-1}{2}=0.02027.
\end{equation} 

\begin{figure}
\includegraphics[height=0.8\hsize, angle=270]{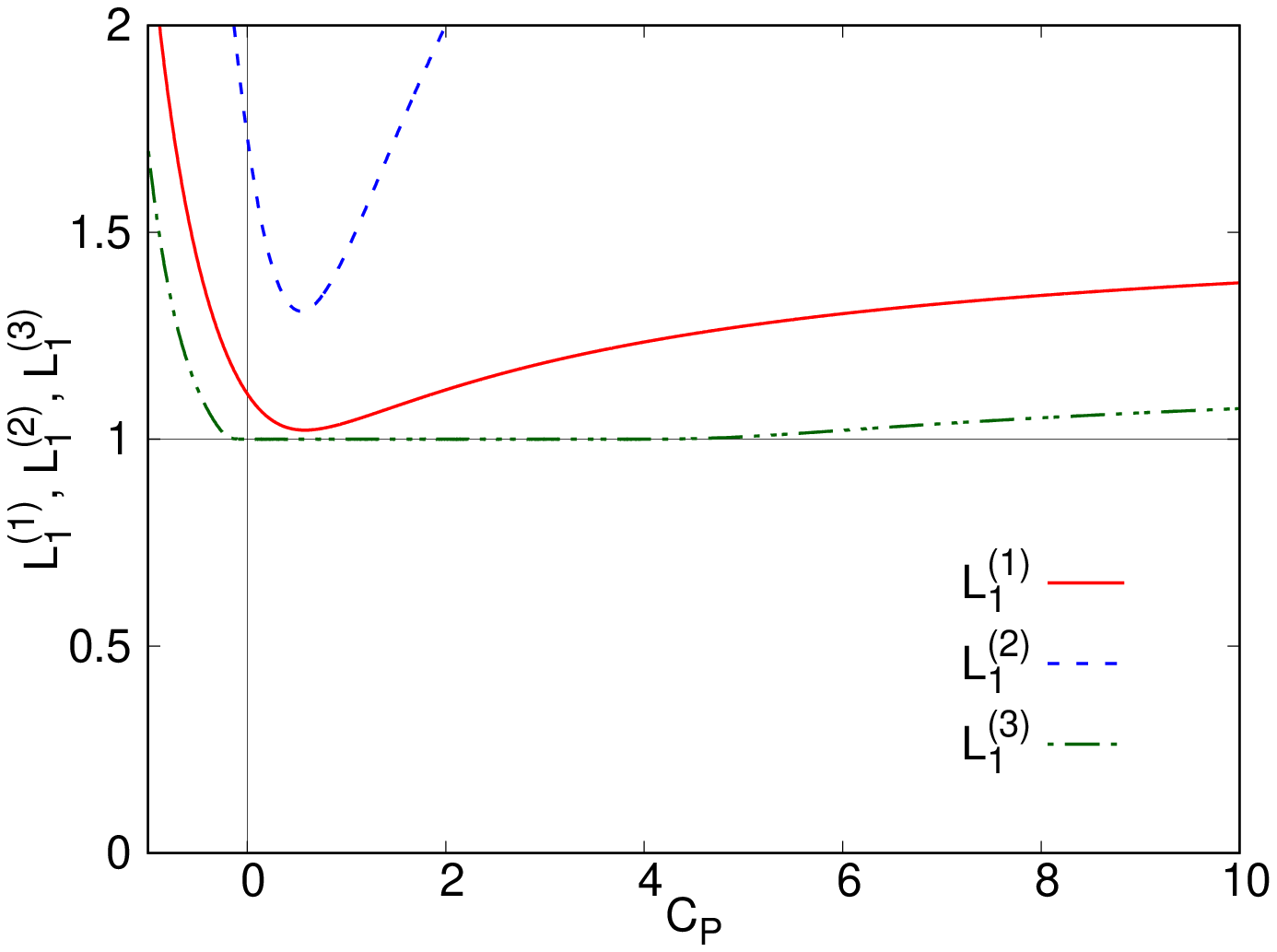}
\caption{Upper bounds for the Schatten--von Neumann $1$-norm of the operator $\widehat{K}$ as a function of the polynomial parameter $C_{P}$. The Gaussian parameters are given as $A=\frac 32$, $C=1$, $B=D=E=0$, and the polynomial parameters are $A_{P}=-1$, $F_{P}=1$, and $B_{P}=D_{P}=E_{P}=0$. Note that the operator $\widehat{K}$ is easily seen to be positive semidefinite at $C_{P}=1$. Here $L_1^{(1)}=\sqrt{R(w_\text{min})}$ is the solid line, $L_1^{(2)}=\sqrt{R(w_\text{min}^G)}$ is the dashed curve. The dash-dotted curve $L_1^{(3)}$ shows an estimate obtained from numerical diagonalization of the operator $\widehat{K}$, for the details see Subsection~\ref{ss:num}. \label{fig:L1norm_as_a_function_gamma2}}
\end{figure}

\section{A monotone approximation of the minimal eigenvalue} \label{s:lmin}

Let $O\colon \iH\to \iH$ be a self-adjoint, trace-class linear operator acting on a separable, complex Hilbert space $\iH$. Then $O$ has countably many real eigenvalues $\{\lambda_i\}_{i\geq 0}$ satisfying $\sum_{i\geq 0} |\lambda_i|<\infty$. If there are only finitely many eigenvalues $\lambda_0,\dots,\lambda_m$ then define $\lambda_i=0$ for each integer $i>m$. From now on let us fix $c>0$ satisfying 
 \begin{equation} \label{eq:c} 
 \sum_{i=0}^{\infty} |\lambda_i|\leq c.
 \end{equation}

\subsection{The estimate \texorpdfstring{$q_{n,c}$}{TEXT} and its speed of convergence} \label{ss:qnc}
\begin{definition}
For $n\geq 0$ define $h_n=h_{n,c}\colon [0,\infty)\to \RR$ as 
\begin{equation*} h_{n,c}(x)\defeq \sum_{k=0}^n e_kx^k-\sum_{k=n+1}^{\infty} \frac{(cx)^k}{k!}, 
%=\sum_{k=0}^n \left(e_k+\frac{c^k}{k!}\right)x^k-e^{cx},
\end{equation*}
and define $q_n=q_{n,c}$ as\footnote{If it does not cause confusion, we simply use $h_n$ and $q_n$ instead of $h_{n,c}$ and $q_{n,c}$, respectively.}
\begin{equation}  \label{eq:q_n_definition} 
q_{n,c}\defeq \frac{-1}{\min\{x>0: h_{n,c}(x)=0\}}.
\end{equation}
\end{definition}

The following theorem shows that $q_n$ is a monotone, asymptotically sharp lower bound of the minimal eigenvalue $\lambda_{\min}$. 

\begin{theorem} \label{t:qmon}
The sequence $q_n$ is monotone increasing and 
\begin{equation*} \lim_{n\to \infty} q_n= \lambda_{\min}.
\end{equation*} 
\end{theorem}

\begin{proof} First we show that for all $x\geq 0$ and $n\geq 0$ we have 
\begin{equation} \label{eq:gn+1}
g(x)\geq h_{n+1}(x)\geq h_n(x). 
\end{equation}
Let us fix $x\geq 0$ and $n\geq 0$. 
The definition of $g,h_n$ and \eqref{eq:ek} imply that 
\begin{equation*}
g(x)=\sum_{k=0}^{\infty} e_k x^k\geq \sum_{k=0}^n e^k x^k-\sum_{k=n+1}^{\infty} \frac{(cx)^k}{k!}=h_n(x).
\end{equation*}
Inequality \eqref{eq:ek} also implies that 
\begin{equation*}
h_{n+1}(x)-h_n(x)=\left(e_{n+1}+\frac{c^{n+1}}{(n+1)!}\right)x^{n+1}\geq 0, 
\end{equation*}
thus \eqref{eq:gn+1} holds.
Therefore the definition of $q_n$, inequalities \eqref{eq:lmin} 
and \eqref{eq:gn+1}, and $g(0)=h_n(0)=1$ imply that $q_n$ is monotone increasing and 
$q_n\leq \lambda_{\min}$ for all $n\geq 0$. Finally, $q_n\to \lambda_{\min}$ follows from the fact that $h_n$ uniformly converges to $g$ on any bounded interval $I\subset [0,\infty)$. 
\end{proof}

The following theorem states that if $\lambda_{\min}=0$ then $q_{n}\approx -\frac cn$, see Figure~\ref{fig:poligauss_poz} for a numerical example. 
\begin{theorem} \label{t:qc} Let $\lambda_{\min}=0$. Then for any integer $n\geq 0$ we have 
\begin{equation*} 
\frac{c}{n+1} \leq |q_n|\leq \frac{ec}{n+1}.
\end{equation*} 
Moreover, the lower bound of $|q_n|$ holds without the assumption $\lambda_{\min}=0$, too. 
\end{theorem} 

\begin{proof} 
First we prove the upper bound. Let $\psi_n(x)=\sum_{k=n+1}^{\infty} \frac{x^{k}}{k!}$ and let $b_n$ be the only positive solution to the equation $\psi_n(x)=1$. As $e_k\geq 0$ for all $k$, we have $\sum_{k=0}^n e_kx^k\geq 1$ for all $x\geq 0$. Therefore, the definition of $h_n$ implies that  
\begin{equation*} \min\{x>0: h_n(x)=0\}\geq \min\{x>0: \psi_n(cx)=1\}=\frac{b_n}{c},
\end{equation*}  
which yields
\begin{equation} \label{eq:qb}
|q_n|\leq \frac{c}{b_n}.
\end{equation}
Before proving a lower bound for $b_n$ we show that 
\begin{equation} \label{eq:Kn}
b_n\leq \frac{n+2}{2}.
\end{equation}
Let $x_n=\frac{n+2}{2}$, we need to prove that $\psi_n(x_n)\geq 1$. 
By a straightforward estimate 
\begin{equation*}  \psi_{n}(x_n)\geq \frac{(x_n)^{n+1}}{(n+1)!}=\frac{(n+2)^{n+1}}{2^{n+1} (n+1)!}, 
\end{equation*} 
so it is enough to show that the sequence $a_n\defeq \frac{(n+2)^{n+1}}{2^{n+1} (n+1)!}$ satisfies $a_n\geq 1$ for all $n$. Indeed, $a_0=1$ and using the monotonicity of the sequence $\left(1+\frac 1n\right)^n$ for all $n\geq 1$ we obtain 
\begin{equation*} \frac{a_n}{a_{n-1}}=\frac 12 \left(1+\frac{1}{n+1}\right)^{n+1}>1, 
\end{equation*} 
so $a_n\geq 1$. Therefore \eqref{eq:Kn} holds. 

By Fact~\ref{f:tail} we have
\begin{equation*} 
\psi_n(x)\leq \frac{2x^{n+1}}{(n+1)!}  \quad \text{for all } 0\leq x\leq \frac{n+2}{2}.
\end{equation*}   
Using this together with \eqref{eq:Kn}, and the well-known inequality $n!\geq e\left(\frac{n}{e}\right)^{n}$ we obtain 
\begin{equation*} 
1=\psi_n(b_n)\leq  2 \frac{(b_n)^{n+1}}{(n+1)!}<  \left(\frac{eb_n}{n+1} \right)^{n+1}. \end{equation*}
Thus $b_n\geq \frac{n+1}{e}$. Hence \eqref{eq:qb} yields
\begin{equation*}  
|q_n|\leq \frac{c}{b_n}\leq \frac{ec}{n+1},
\end{equation*}
which finishes the proof of the upper bound. 

Finally, we will prove the lower bound. For integers $n\geq 0$ and reals $x\geq 0$ consider the function 
\begin{equation*} 
f_n(x)=\sum_{k=0}^n \frac{x^k}{k!}-\sum_{k=n+1}^{\infty} \frac{x^k}{k!}, 
\end{equation*} 
and let 
\begin{equation*} 
d_n=\min\{x>0: f_n(x)=0\}.
\end{equation*} 
By \eqref{eq:ek} we obtain that $h_n(x)\leq f_n(cx)$. This inequality and $h_n(0)=f_n(0)=1$ imply that 
\begin{equation*}
\min\{x>0: h_n(x)=0\}\leq \min\{x>0: f_n(cx)=0\} \leq \frac{d_n}{c}, 
\end{equation*}
which yields that 
\begin{equation} \label{eq:qd}
|q_n|\geq \frac{c}{d_n}. 
\end{equation}
Thus it is enough to show that 
\begin{equation} \label{eq:dn} 
d_n\leq n+1,
\end{equation} 
for which it is sufficient to see that $f_n(n+1)\leq 0$. In order to see this, it is enough to prove that for all integers $0\leq m\leq n$ we have 
\begin{equation} \label{eq:frac} \frac{(n+1)^{n-m}}{(n-m)!}\leq \frac{(n+1)^{n+1+m}}{(n+1+m)!},
\end{equation} 
since summarizing \eqref{eq:frac} from $m=0$ to $n$ clearly yields 
\begin{equation*} 
\sum_{k=0}^n \frac{(n+1)^k}{k!}\leq \sum_{k=n+1}^{2n+1} \frac{(n+1)^k}{k!}<\sum_{k=n+1}^{\infty} \frac{(n+1)^k}{k!}.
\end{equation*} 
Inequality~\eqref{eq:frac} is equivalent to 
\begin{equation*}
\prod_{i=-m}^{m} (n+1+i)\leq (n+1)^{2m+1},
\end{equation*}
which clearly follows from the identity 
\begin{equation*} 
(n+1-i)(n+1+i)=(n+1)^2-i^2\leq (n+1)^2.
\end{equation*}
Hence we proved \eqref{eq:dn}. Inequalities \eqref{eq:qd} and \eqref{eq:dn} imply that 
\begin{equation*}
|q_n|\geq \frac{c}{d_n}\geq \frac{c}{n+1},  
\end{equation*}
so the proof of the lower bound is also complete. 
\end{proof}

\begin{figure}
\includegraphics[height=0.8\hsize, angle=270]{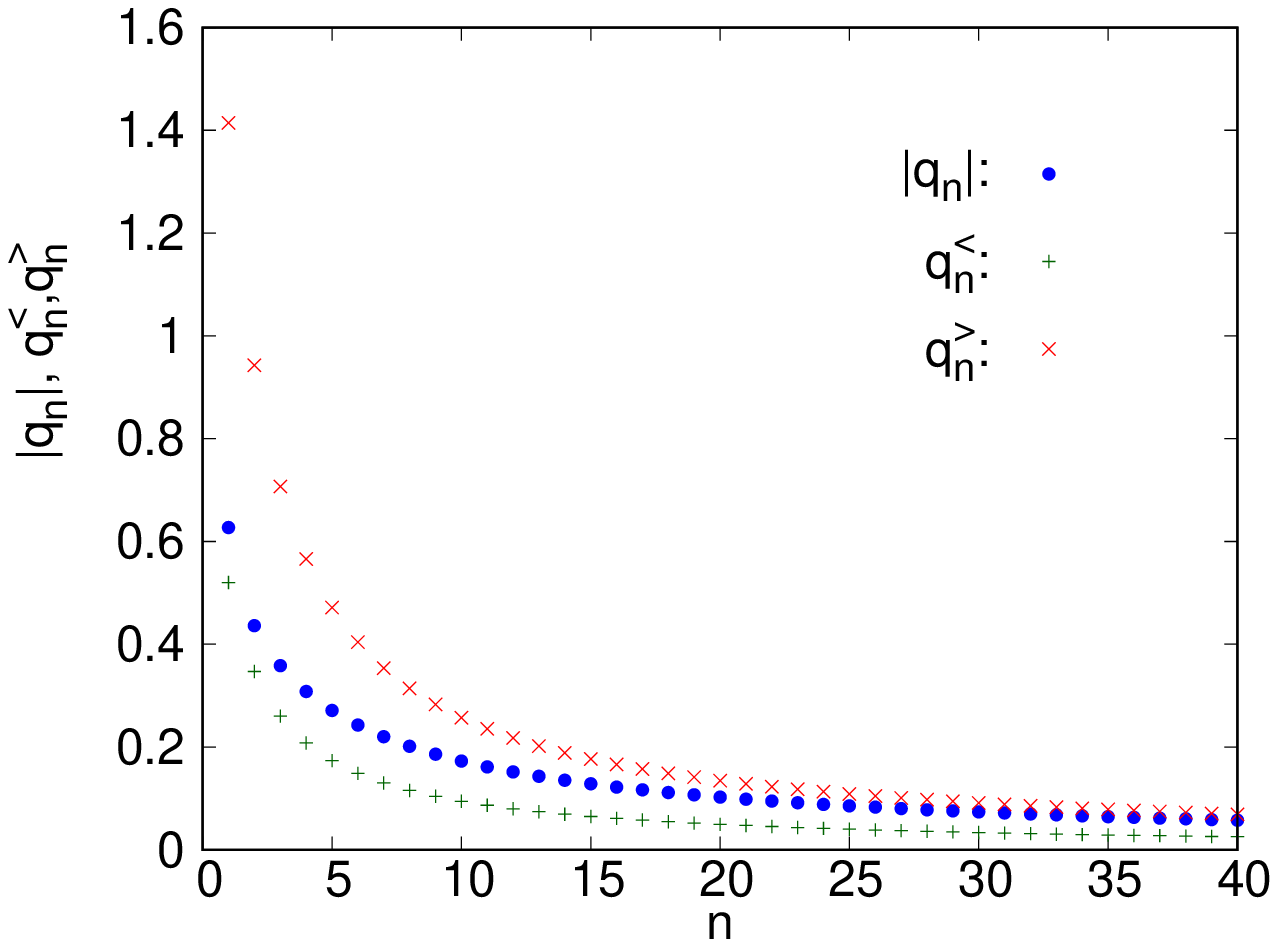}
\caption{Consider the positive semidefinite polynomial Gaussian operator $\widehat{K}$ of the form \eqref{eq:rho_Q} with parameters $A=\frac 32$, $C=1$, $B=D=E=0$, $A_P=-1$, $C_P=1$, $F_P=1$, and $B_P=D_P=E_P=0$. By \eqref{eq:104} the upper bound of $\big\|\widehat{K} \big\|_1$ can be chosen to be $c=1.04054$. We plotted $|q_n|=|q_{n,c}|$, $q_n^<=\frac{c}{n+1}$, and $q_n^>=\frac{ec}{n+1}$ as a function of $n$ calculated from \eqref{eq:q_n_definition}.
\label{fig:poligauss_poz}}
\end{figure}

In the case of $\lambda_{\min}<0$ the following theorem establishes that $q_n\to \lambda_{\min}$ with super-exponential speed. On the good side,  $q_{n,c}$ is monotone, so our method provides a rigorous lower bound for $\lambda_{\min}$. On the bad side, $q_{n,0}$ approximates $\lambda_{\min}$ much better than $q_{n,c}$, and the main reason of this is that $q_{n,0}$ does not depend on the value of $c$ through the sub-optimal estimate \eqref{eq:ek}.  Note that inequality \eqref{eq:pn} is better than \eqref{eq:delta}, since the coefficient of $-n\log(n)$ in the exponent is $1$ instead of $\frac{|\lambda_{\min}|}{c}\leq 1$, avoiding the critical case when $\frac{|\lambda_{\min}|}{c}$ is very small.

\begin{figure}
\includegraphics[height=0.8\hsize, angle=270]{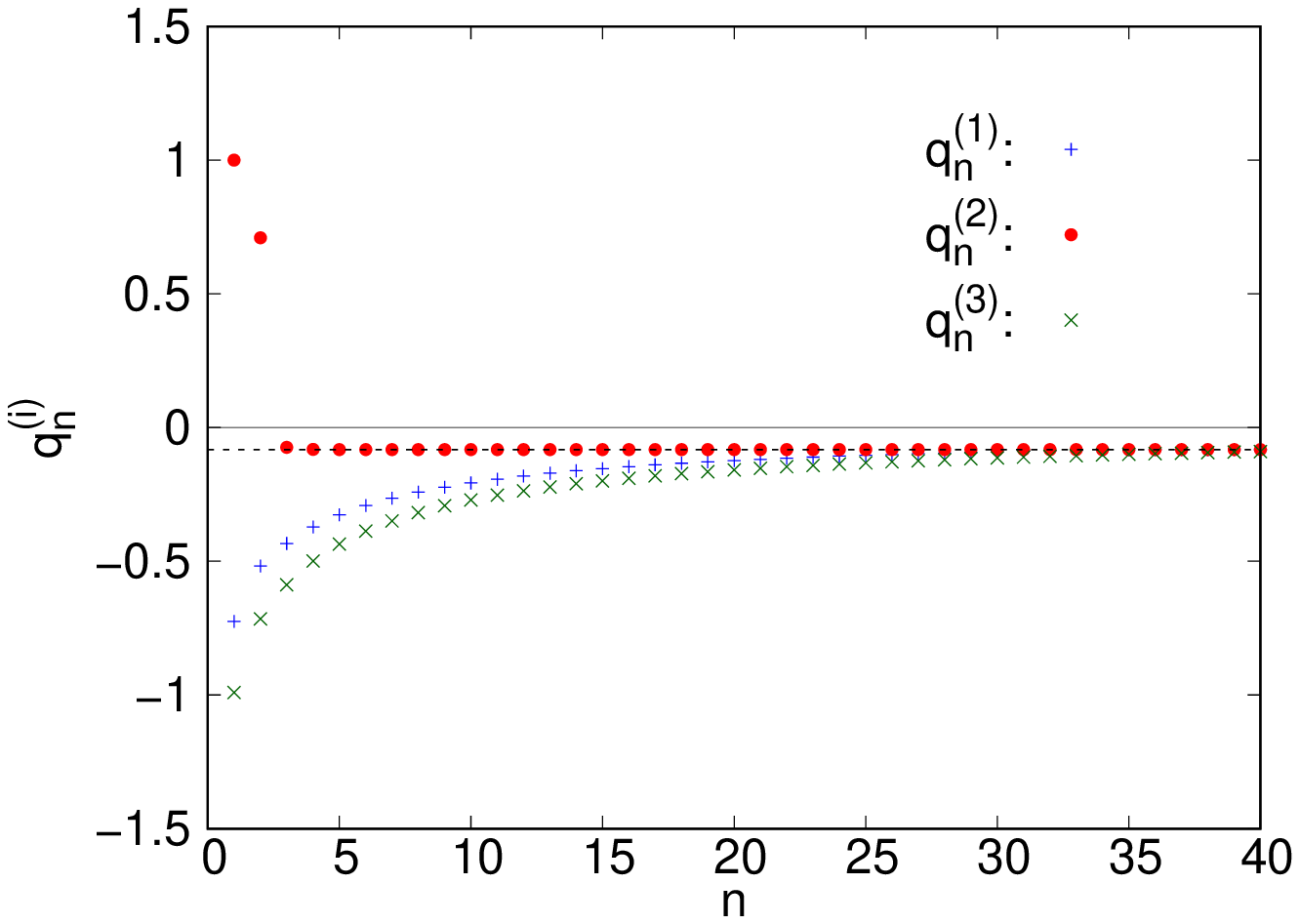}
\caption{Consider the polynomial Gaussian operator $\widehat{K}$ of the form \eqref{eq:rho_Q} with parameters $A=\frac 32$, $C=1$, $B=D=E=0$, $A_P=-1$, $C_P=40$, $F_P=1$, and $B_P=D_P=E_P=0$. We plotted $q_n^{(i)}=q_{n,c}$ as a function of $n$ calculated from \eqref{eq:q_n_definition} and \eqref{eq:qn0}.  
For $i=1$ we obtain $c=\sum_j |\lambda_j|=1.16445$ from diagonalization, for the details see Subsection~\ref{ss:num}. 
For $i=2$ we take $c=0$, and for $i=3$ the value $c=1.4941$ was calculated from the H\"older's upper bound and optimization given in Subsection~\ref{ss:quad}. The horizontal dashed line indicates the common limit  $-0.082228$.}
\label{fig:nongaussian_qns}
\end{figure}

 \begin{theorem} \label{t:qexp}
Let $\lambda_{\min}<0$ and $\alpha=\frac{c}{|\lambda_{\min}|}$. Then for all $n$ we have 
\begin{equation} \label{eq:qlambda} |q_{n,c}-\lambda_{\min}|\leq \frac{\delta_n}{1-\delta_n} |\lambda_{\min}|,
\end{equation} 
where $0<\delta_n<1$ is the unique solution of the equation  
\begin{equation*} x^{\alpha}=2\sum_{k=n+1}^{\infty} \frac{(\alpha(1-x))^{k}}{k!}.
\end{equation*} 
For $n\geq \max\{3,2(\alpha-1)\}$ we have the estimate 
\begin{equation} \label{eq:delta}  
\delta_n\leq \left(\frac{e\alpha}{n+1}\right)^{\frac{n+1}{\alpha}}=\exp \left[-\frac{|\lambda_{\min}|}{c} n \log(n)+\mathcal{O}(n)\right]. 
\end{equation} 
\end{theorem} 

\begin{proof}
First we show \eqref{eq:qlambda}. Let 
\begin{equation*} 
\theta=\min\{x>0: g(x)=0\}=\frac{1}{|\lambda_{\min}|} \quad \text{and} \quad   
\theta_n=\min\{x>0: h_n(x)=0\}. 
\end{equation*} 
Since $\theta_n\leq \theta$, we can define $0\leq s_n<1$ such that $\theta_n=(1-s_n) \theta$. Then we have 
\begin{equation*}
|q_n-\lambda_{\min}|=\left|-\frac{1}{\theta_n}+\frac{1}{\theta}\right|=\frac{|\theta_n-\theta|}{\theta_n \theta}=\frac{s_n}{1-s_n} |\lambda_{\min}|, 
\end{equation*}
so as the function $x\mapsto \frac{x}{1-x}$ is monotone increasing on $[0,1)$, it is enough to show that $s_n\leq \delta_n$. 
For $0\leq x\leq 1$ let  
\begin{equation*}  
\phi_n(x)=x^{\alpha}-2\sum_{k=n+1}^{\infty} \frac{(\alpha(1-x))^{k}}{k!}. 
\end{equation*} 
Since $\phi_n$ is strictly monotone increasing and $\phi_n(\delta_n)=0$, it is enough to prove that 
\begin{equation} \label{eq:phi} 
\phi_n(s_n)\leq 0.
\end{equation} 
First we prove that 
\begin{equation} \label{eq:theta} 
g(\theta_n)\geq s_n^{\alpha}.
\end{equation} 
Let $\lambda=|\lambda_{\min}|$. It is easy to see that the function 
$x\mapsto (1-x)^{\frac 1x}$ is monotone decreasing on $(0,1]$, so for all $0\leq x\leq \theta =\frac{1}{\lambda}$ we have 
\begin{equation*} 
g(x)=\prod_{i=0}^{\infty} (1+\lambda_ix)\geq \prod_{i: \lambda_i<0} (1+\lambda_i x)\geq  \prod_{i: \lambda_i<0}  (1-\lambda x)^{\frac{|\lambda_i|}{\lambda}}\geq (1-\lambda x)^{\alpha}, 
\end{equation*} 
which implies \eqref{eq:theta}. On the other hand, by the definitions of $g,h_n$ and \eqref{eq:ek} for all $x\geq 0$ we obtain 
\begin{equation*} 
g(x)-h_n(x)=\sum_{k=n+1} ^{\infty} \left(e_k+\frac{c^k}{k!}\right)x^k\leq 2 \sum_{k=n+1} ^{\infty} \frac{(cx)^k}{k!},
\end{equation*} 
therefore
\begin{equation} \label{eq:sn} g(\theta_n)\leq 2\sum_{k=n+1}^{\infty} \frac{(\alpha(1-s_n))^{k}}{k!}.
\end{equation}
Then \eqref{eq:theta} and \eqref{eq:sn} imply \eqref{eq:phi}, so the proof of \eqref{eq:qlambda} is complete. 

Finally, we prove \eqref{eq:delta}. By a well-known estimate for factorials we have 
\begin{equation} \label{eq:fact}  
(n+1)!\geq \sqrt{2n\pi} \left( \frac{n+1}{e}\right)^{n+1}> 4 \left( \frac{n+1}{e}\right)^{n+1}
\end{equation} 
and $n\geq \max\{3,2(\alpha-1)\}$ implies
\begin{equation} \label{eq:n+2} 
 \alpha(1-\delta_n)<\alpha\leq \frac{n+2}{2}.
\end{equation} 
Then \eqref{eq:n+2} with Fact~\ref{f:tail} and \eqref{eq:fact} yield
\begin{equation*} \delta_n^{\alpha} =2\sum_{k=n+1}^{\infty} \frac{(\alpha(1-\delta_n))^{k}}{k!}\leq 4\frac{(\alpha (1-\delta_n))^{n+1}}{(n+1)!} \leq \left(\frac{e\alpha (1-\delta_n)}{n+1}\right)^{n+1}.
\end{equation*} 
Substituting $N=\frac{n+1}{\alpha}$ implies 
\begin{equation*} N\delta_n^{\frac 1N}\leq e(1-\delta_n)\leq e.
\end{equation*}  
Thus
\begin{equation*} \delta_n\leq \left(\frac eN\right)^N=\left(\frac{e\alpha}{n+1}\right)^{\frac{n+1}{\alpha}}=\exp \left[-\frac{|\lambda_{\min}|}{c} n \log(n)+\mathcal{O}(n)\right], 
\end{equation*} 
hence \eqref{eq:delta} holds. This completes the proof of the theorem. 
\end{proof}

\subsection{Numerical calculation of the spectrum in an important special case} \label{ss:num}
Matrix representation $K_{m,n}$ of an arbitrary quadratic polynomial Gaussian kernel in case of $d=1$ can be calculated rigorously, see \cite[Equations~(22),(23)]{Newton}. Since the matrix elements $K_{m,n}$ decrease exponentially as $m,n\to \infty$, we can truncate the matrix at a finite size, and the original spectrum can be numerically approximated with the eigenvalues of the resulting finite dimensional operator. In the more general case, numerical diagonalization is not an easy task even for a polynomial Gaussian kernel. Our results derived from the numerical diagonalization are shown in Figure~\ref{fig:L1norm_as_a_function_gamma2}, and we also used it for Figure~\ref{fig:nongaussian_qns}. The efficiency and accuracy of this method is confirmed by comparing it with the results of \cite{Newton,balka2024positivity}.

\section{Discussion and conclusions}

The eigenvalues of self-adjoint trace-class operators play a very important role in physics. These operators naturally appear in many areas of quantum physics, and their eigenvalues are given by a Fredholm type integral equation, whose analytical solution is unfortunately completely hopeless in general. As our main result, we approximated the spectrum of the aforementioned operators both in theory and in practice. In particular, our methods provide  estimates of the Schatten--von Neumann 1-norm and the von Neumann entropy. 

We also investigated a mathematically precise non-monotone approximation of the minimal and maximal eigenvalues, which also often play an important role in physics. Here, we paid special attention to estimating the speed of convergence, which is important for many practical applications. We proved that if the there is a negative eigenvalue, then the speed of convergence to the minimal eigenvalue given by our algorithm is super-exponential, see Section~\ref{s:qn0}. Analogously, the speed of convergence to the maximal eigenvalue is always super-exponential.

In Section~\ref{s:Sch} we provided rigorous upper bounds for the Schatten–von Neumann 1-norm. We demonstrated that it works effectively for Gaussian and polynomial Gaussian operators as well. It is interesting that non self-adjoint operators also come into the picture: we decompose our operator to the product two not necessarily self-adjoint operators, and bound the $1$-norm of our operator by calculating the two $2$-norms and using H\"older's inequality. This method might be advantageous if calculating the moments of our operator is challenging; and the aforementioned decompositions might be useful in practice as well.  

%It may happen in some cases that the estimation of the Schatten--von Neumann based on the method described in Section \ref{s:spectrum} is complicated in practice. In such cases, this method based on mathematically rigorous estimation can bring results faster. Although this gives a completely accurate estimate for the Schatten–von Neumann 1-norm only in certain cases. The knowing of the relatively inaccurate upper bound is  very useful for example  during the monotonic approximation of the minimal and maximal eigenvalues.

Often, the only goal is to decide whether a self-adjoint trace-class operator is positive semidefinite, which depends on whether the minimal eigenvalue is negative. We constructed a monotone sequence of lower bounds for the minimal eigenvalue that only depends on the moments of the operator and a concrete upper estimate of its Schatten--von Neumann $1$-norm, and also converges to the minimal eigenvalue. We demonstrated that this approximation can be effectively calculated for a large class of physically relevant operators. We proved that this approximation converges with super-exponential speed to the minimal eigenvalue if the operator is not positive semidefinite.

Finally, we note that the algorithm outlined in this paper is an efficient alternative of the  numerical diagonalization  for a polynomial
Gaussian kernel, see Subsection~\ref{ss:num}. In principle, it can be applied to a much broader class of self-adjoint operators, not only to polynomial Gaussian operators.  A major advantage of our method is that it avoids the need for an appropriate basis for diagonalization. However, it requires an extreme precise determination of the moments $M_k$. Usually, the calculation of low moments is easy, but determination of higher moments can be more demanding. 
The next step, the calculation of $e_k$ from the moments can be also difficult numerically. The numbers $e_k$ converge much faster to zero than the moments $M_k$ for increasing $k$. With extended precision to a few tens of digits, offered by Mathematica, our method presented here is numerically stable, accurate, fast, and it can be implemented efficiently. In Appendix~\ref{app:B}, we quantify this numerical stability by studying how errors in the approximated moments propagate to the approximations of the coefficients $e_n$.

%Finally, we note that this very effective method can be generalized to the case of vector variables and polynomials of higher degrees, too.

%Végül megjegyezzük, hogy az ebben a cikkben felvázolt algoritmus még ez utóbbinál is effektívebb, valamint sokkal általánosabb osztályára alkalmazható az önadjungált  operátoroknak, nemcsak polinomszor gauss alakú operátorokra. Nagy előnye, hogy mindezt diagonalizálás nélkül tehetjük meg és numerikusan is jól és stabilan, kellő pontossággal és gyorsasággal programozható.

\appendix 
\section{Approximating \texorpdfstring{$\Lambda$}{Lambda}}
\label{app:A}

The goal of this appendix is to generalize Theorem~\ref{t:conv} to arbitrary trace-class operators, removing the assumptions of self-adjointness and simplicity of the nonzero eigenvalues.

\begin{definition}
For any integer $n\geq 1$ define the set $\Gamma_n\subset \CC$ as
\begin{equation*} 
\Gamma_n= \left\{-\frac{1}{x}: x\in \CC \text{ and } h_n(x)=0\right\}, \text{ where } h_n(x)=\sum_{k=0}^n e_kx^k.
\end{equation*} 
\end{definition}
Define $h\colon \CC\to \CC$ as 
 \begin{equation*} 
 h(x)\defeq \prod_{i=0}^{\infty} (1+\lambda_i x),
 \end{equation*}
where the eigenvalues $\lambda_i\in \CC$ satisfy $\sum_{i=0}^{\infty} |\lambda_i|<\infty$. 
Then \cite[Lemma~3.3]{simon} yields that $h(x)$ is finite for all $x\in \CC$ and
 \begin{equation*} h(x)=\sum_{k=0}^{\infty} e_k x^k.
 \end{equation*}
Let us recall
\begin{equation*}
\Lambda=\overline{\{\lambda_i: i\geq 0\}\setminus\{0\}}. 
\end{equation*}
For the definition of the Hausdorff metric on $\CC$ see Definition~\ref{d:Haus} for $d=2$. If $\Lambda=\emptyset$, then $h_n=h\equiv 1$ for all $n$, so we may assume that this is not the case.  
\begin{theorem}
If $\Lambda\neq \emptyset$ then $\Gamma_n \to \Lambda$ in Hausdorff metric.
\end{theorem}

\begin{proof} 
First assume that there are only finitely many non-zero $\lambda_i$, say $\lambda_0, \dots, \lambda_{k-1}$. Then clearly $h_n(x)=h(x)$ for all $n\geq k$ for each $x$. Hence $\Gamma_n=\Lambda$ for all $n\geq k$.  

Now assume that there are infinitely many non-zero $\lambda_i$. Fix an arbitrary $\varepsilon>0$, we need to prove that there exists $N\in \NN$ such that $\dH(\Gamma_n, \Lambda)\leq \varepsilon$ for all $n\geq N$. As $0\in \Lambda$, we can define 
\begin{equation*} 
K_{\varepsilon}=\left\{-\frac 1x: x\in \CC\setminus B(\Lambda,\varepsilon)\right\}\cup \{0\}.
\end{equation*}
Clearly $K_{\varepsilon}$ is closed, and $K_{\varepsilon}\subset \overline{B(0,1/\varepsilon)}$, so $K_{\varepsilon}$ is compact. As $h$ is continuous and non-zero on $K_{\varepsilon}$, we can define $\theta\defeq \min\{|h(x)|: x\in K_{\varepsilon}\}>0$. Since $h_n\to h$ uniformly on $K_{\varepsilon}$, there exists $N_1\in \NN$ such that $|h_n(x)-h(x)|<\theta$ for all $n\geq N_1$ and $x\in K_{\varepsilon}$. Therefore, the definition of $\theta$ yields $K_{\varepsilon}\cap \{x: h_n(x)=0\}=\emptyset$ for all $n\geq N_1$. Hence the definitions of $K_{\varepsilon}$ and $\Gamma_n$ imply for each $n\geq N_1$ that
\begin{equation} \label{eq:G1}
\Gamma_n\subset B(\Lambda,\varepsilon).
\end{equation}
By the compactness of $\Lambda$ and the density of the nonzero eigenvalues in it, we can choose finitely many non-zero $\lambda_i$ (which we may relabel as $\lambda_0,\dots, \lambda_{k-1}$) and $0<\delta<\min_{0\leq i<k} |\lambda_i|$ such that for any points $x_i\in B(\lambda_i,\delta)$ we have 
\begin{equation} \label{eq:epsd} 
\Lambda\subset B\left(\cup_{i=0}^{k-1} \{x_i\},\varepsilon\right).
\end{equation} 
Also choose $r_i>0$ such that for all $0\leq i<k$ we have 
\begin{equation} \label{eq:ri} B\left(-1/\lambda_i,r_i\right)\subset \{-1/z: z\in B(\lambda_i,\delta)\}.
\end{equation}

Recall that $h(-1/\lambda_i)=0$ for all $0\leq i<k$.
As the entire functions $h_n$ tend uniformly on compact subsets to the function $h\not\equiv 0$, Hurwitz's theorem \cite[p.~178, Theorem~2]{Ahlfors} implies that there exists $N_2\in \NN$ such that for all $n\geq N_2$ for all $0\leq i<k$ there are $z_{i,n}\in B(-1/\lambda_i,r_i)$ such that $h_n(z_{i,n})=0$.  Let $x_{i,n}=-1/z_{i,n}$, then clearly $\{x_{i,n}: 0\leq i<k\}\subset \Gamma_n$, and \eqref{eq:ri} yields that  $x_{i,n}\in B(\lambda_i,\delta)$ for all $0\leq i<k$. Then \eqref{eq:epsd} implies that 
\begin{equation} \label{eq:G2}
\Lambda\subset B(\{x_{i,n}: 0\leq i<k\},\varepsilon) \subset B(\Gamma_n,\varepsilon).
\end{equation} 
Finally, let $N=\max\{N_1,N_2\}$, then \eqref{eq:G1} and \eqref{eq:G2} imply for all $n\geq N$ that
\begin{equation*}
\dH(\Gamma_n,\Lambda)\leq \varepsilon. 
\end{equation*} 
As $\varepsilon>0$ was arbitrary, the proof is complete. 
\end{proof}

\section{Stability of the approximation}
\label{app:B}

Let $\widehat{e}_n$ be the approximation of the $n$th elementary symmetric polynomial of the eigenvalues, which is calculated from the approximated moments $\widehat{M}_1,\dots,\widehat{M}_n$ by the Plemelj--Smithies formulas \eqref{eq: e_k-k}. Let $M_k$ denote the exact moments, and assume that for some fixed $0<\eta \leq 1$ we have \begin{equation} 
\left|\widehat{M}_k-M_k\right|\leq \eta \quad \textrm{for all } 1\leq k\leq n. 
\end{equation}
Finally, let $C\defeq \sum_{i=0}^{\infty} |\lambda_i|<\infty$. The main goal of this Appendix is to determine how the error $\left|\widehat{e}_n-e_n\right|$ depends on $\eta$.
\begin{theorem} \label{t:up}
We have 
\begin{equation*}
\left|\widehat{e}_n-e_n\right|\leq \sum_{k=1}^{n} \frac{(\eta)_k}{k!}|e_{n-k}|\leq \left(e^{C+1}+o_{C,\eta}(1)\right)\frac{\eta}{n^{1-\eta}} 
\end{equation*}
as $n\to \infty$, where 
\begin{equation*} (\eta)_k \defeq\eta(\eta+1)\cdots (\eta+k-1).
\end{equation*} 
\end{theorem}
\begin{corollary}
Suppose that the first $n$ moments are calculated to $m$ decimal places, that is, their absolute errors are at most $\eta=10^{-m}$.
Applying our theorem with this $\eta$, we obtain the upper bound \begin{equation*} \left|\widehat{e}_n-e_n\right|\leq \left(e^{C+1}+o_{C,\eta}(1)\right)\frac{\eta}{n^{1-\eta}}. 
\end{equation*} 
Thus, for fixed $C$ and $m$ as $n\to \infty$, $\widehat{e}_n$ has at least 
\begin{equation*} 
 m+(1-10^{-m})\log_{10}n-\mathcal{O}_{C,\eta}(1)
 \end{equation*}
correct decimal places in absolute accuracy. For large $m$, this corresponds to a gain of approximately
$\log_{10}n$ decimal places in absolute accuracy.
A relative error bound is obtained by dividing the absolute error bound by $|e_n|$.
As $|e_n|\le C^n/n!$ decays faster than exponentially, this bound may not guarantee useful relative accuracy even for moderately large $n$.
\end{corollary}

\begin{proof}[Proof of Theorem~\ref{t:up}]
Fix $n$ and extend the approximated moments by setting $\widehat M_j=M_j$ for $j>n$.
Define $\widehat e_k$ for all $k\ge0$ by the Plemelj--Smithies formulas,
where $\widehat e_0=1$. This leaves $\widehat e_0,\ldots,\widehat e_n$
unchanged, since $\widehat e_k$ depends only on the first $k$ moments.
Put $\Delta_j=\widehat M_j-M_j$. Then $|\Delta_j|\le\eta$ for $j\le n$
and $\Delta_j=0$ for $j>n$.

Expanding the determinant in the Plemelj--Smithies formula \eqref{eq: e_k-k} along the last row gives
\begin{equation} \label{e:ke}
ke_k=\sum_{j=1}^k(-1)^{j-1}M_je_{k-j}.
\end{equation}
Recall that $g(x)=\sum_{k=0}^{\infty} e_k x^k$. Thus multiplying \eqref{e:ke} by $x^{k-1}$
and summing over $k\geq 1$ yields 
\begin{equation*} 
g'(x)=g(x)\sum_{j=1}^\infty(-1)^{j-1}M_jx^{j-1}.
\end{equation*} 
Since $g(0)=1$, it follows that
\begin{equation} \label{e:g} 
g(x)=\exp\left(\sum_{j=1}^\infty\frac{(-1)^{j-1}M_j}{j}x^j\right).
\end{equation} 
The same argument applies to the approximated moments and their corresponding determinants, giving
\begin{equation} \label{e:g2} 
\widehat{g}(x) \defeq \sum_{k=0}^\infty\widehat{e}_kx^k
=\exp\left(\sum_{j=1}^\infty\frac{(-1)^{j-1}\widehat{M}_j}{j}x^j\right).
\end{equation} 
These generating-function identities may be read as identities of formal power
series; they also converge for sufficiently small $|x|$, since the moments grow
at most exponentially. Consequently, \eqref{e:g} and \eqref{e:g2} imply that
\begin{equation} \label{e:ghat}
\widehat{g}(x)=g(x)\exp\left(\sum_{j=1}^n
\frac{(-1)^{j-1}\Delta_j}{j}x^j\right).
\end{equation}
Write the exponential factor as
\begin{equation*}
A(x)\defeq \exp\left(\sum_{j=1}^n\frac{(-1)^{j-1}\Delta_j}{j}x^j\right)
=\sum_{k=0}^\infty a_kx^k,
\end{equation*}
we claim that 
\begin{equation} \label{e:ind}
|a_k|\leq \frac{(\eta)_k}{k!} \quad \textrm{for every } k \ge 0.
\end{equation} 
Differentiating and comparing coefficients gives
\begin{equation*}
ka_k=\sum_{j=1}^k(-1)^{j-1}\Delta_ja_{k-j},
\end{equation*}
where $\Delta_j=0$ for $j>n$. Thus
\begin{equation*}
k|a_k|\le\eta\sum_{j=0}^{k-1}|a_j|.
\end{equation*}
Let  $c_0=1$ and $c_k=(\eta)_k/k!$ for $k>0$. The definition gives $jc_j=(\eta+j-1)c_{j-1}$, so 
\begin{equation*} 
jc_j-(j-1)c_{j-1}=\eta c_{j-1}.
\end{equation*} 
Summing the above identity over $j=1,\ldots,k$ gives
\begin{equation} \label{e:kck} 
kc_k=\eta\sum_{j=0}^{k-1}c_j.
\end{equation} 
We prove \eqref{e:ind} by induction on $k$. Clearly $|a_0|=c_0=1$. For $k\geq 1$, assume that $|a_i|\leq c_i$ for all $0\leq i<k$. Then \eqref{e:kck} yields
\begin{equation*} 
k|a_k|\le\eta\sum_{j=0}^{k-1}|a_j|
\le\eta\sum_{j=0}^{k-1}c_j=kc_k,
\end{equation*}
which proves \eqref{e:ind}. 
Since $a_0=1$, comparing the coefficients on both sides of \eqref{e:ghat} yields
\begin{equation} \label{e:ineq1} 
|\widehat e_n-e_n|
=\left|\sum_{k=1}^na_ke_{n-k}\right|
\le\sum_{k=1}^n\frac{(\eta)_k}{k!}|e_{n-k}|.
\end{equation}

Finally, we will bound the sum \eqref{e:ineq1} from above. Since $|e_k|\le \frac{C^k}{k!}$ by \eqref{eq:ek}, it remains to prove the bound
\begin{equation} \label{e:ineq2}
S_n \defeq \sum_{k=0}^{n-1}\frac{C^k}{k!}c_{n-k}\leq \left(e^{C+1}+o_{C,\eta}(1)\right)\frac{\eta}{n^{1-\eta}}.
\end{equation} 
For $k\ge1$, using $1+t\le e^t$ and
$\sum_{r=1}^{k-1}(1/r)\leq 1+\log k$, we have
\begin{equation*} 
c_k=\frac{\eta}{k}\prod_{r=1}^{k-1}\left(1+\frac{\eta}{r}\right)
\le\frac{\eta}{k}\exp\left(\eta\sum_{r=1}^{k-1}\frac1r\right)
\le e^\eta\eta k^{\eta-1}.
\end{equation*} 
Split the sum defining $S_n$ at $k=n/2$. If $k\le n/2$, then $n-k\ge n/2$, and $0<\eta\le 1$ implies that
\begin{equation*} 
c_{n-k}\le e^\eta2^{1-\eta}\eta n^{\eta-1}
\le e\eta n^{\eta-1}.
\end{equation*} 
Here $e^\eta2^{1-\eta}\le e$, because its logarithm
$\eta+(1-\eta)\log2$ is at most $1$. Consequently,
\begin{equation} \label{e:bnd1}
\sum_{0\le k\leq n/2}\frac{C^k}{k!}c_{n-k}
\le e\eta n^{\eta-1}\sum_{j=0}^{\infty}\frac{C^j}{j!}
=e^{C+1}\frac{\eta}{n^{1-\eta}}.
\end{equation} 
For $k>n/2$, we use $c_{n-k}\leq e\eta$, so 
\begin{equation} \label{e:bnd2}
\sum_{n/2<k<n}\frac{C^k}{k!}
\le 2^{-n/2}\sum_{j=0}^{\infty}\frac{(2C)^j}{j!}
=2^{-n/2}e^{2C}.
\end{equation}
Combining the bounds \eqref{e:bnd1} and \eqref{e:bnd2} implies 
\begin{equation*} 
S_n \le e^{C+1}\frac{\eta}{n^{1-\eta}}
+e^{2C+1}\eta\,2^{-n/2}=\left(e^{C+1}+o_{C,\eta}(1)\right)\frac{\eta}{n^{1-\eta}}.
\end{equation*} 
This proves \eqref{e:ineq2}, and the proof is complete. 
\end{proof}

\subsection*{Acknowledgments} We are indebted to J\'ozsef Zsolt Bern\'ad, Andr\'as Bodor, \\ Andr\'as Frigyik, M\'aty\'as Koniorczyk, Mikl\'os Pintér, and G\'eza T\'oth for some helpful conversations. We also thank an anonymous reviewer whose questions led us to add the two appendices. The authors acknowledge the use of ChatGPT (OpenAI) in deriving the error estimate in Appendix B and in making some minor revisions to the manuscript.

\subsection*{Funding} The first author was supported by the National Research, Development and Innovation Office -- NKFIH, grants no.~124749 and 146922, and by the J\'anos Bolyai Research Scholarship of the Hungarian Academy of Sciences. 

The second author thanks the 2021 Thematic Excellence Programme of the NKFIH (Project No.~TKP2021-NVA-04), which has been implemented with the support provided by the Ministry of Innovation and Technology of Hungary from the National Research, Development and Innovation Fund, financed under the TKP2021-NVA funding scheme. 

The third author was supported by the Hungarian National Research, Development and Innovation Office within the Quantum Information National Laboratory of Hungary grants no.~2022-2.1.1-NL-2022-00004 and 134437.

%\subsection*{Declaration and conflict of interest} We confirm that the manuscript has not published elsewhere and is not under consideration by another journal. The authors have no conflict of interest to declare. 

\bibliographystyle{unsrt}
\bibliography{entanglement}

\end{document}